\documentclass{amsproc}

\usepackage{bm} 
\usepackage{url}
\usepackage{caption}
\usepackage{subcaption}
\usepackage{thmtools}
\usepackage{thm-restate}
\def\BibTeX{{\rm B\kern-.05em{\sc i\kern-.025em b}\kern-.08em
    T\kern-.1667em\lower.7ex\hbox{E}\kern-.125emX}}

\usepackage{hyperref}
\usepackage[colorinlistoftodos]{todonotes}

\usepackage{xspace}
\usepackage{algorithmicx}
\usepackage{algorithm}
\usepackage{algpseudocode}
\usepackage[shortlabels]{enumitem}

\usepackage{tikz}
\usetikzlibrary{calc}
\usetikzlibrary{shapes,decorations.markings,arrows.meta,fit}

\usepackage[normalem]{ulem}

\usepackage{microtype}

\usepackage{qtree}

\usepackage{relsize}

\newcommand{\rai}{RelationalAI}

\newcommand{\norm}[1]{\|#1\|}
\newcommand{\set}[1]{\{#1\}}                    
\newcommand{\setof}[2]{\{{#1}\mid{#2}\}}        
\newcommand{\dom}{\textsf{Dom}}

\newcommand{\degree}{\text{\sf deg}}
\newcommand{\polylog}{\text{\sf polylog}}

\newcommand{\opt}{\text{\sf opt}}

\newcommand{\calM}{\mathcal M}
\newcommand{\calS}{\mathcal S}

\newcommand{\calH}{\mathcal H}
\newcommand{\calV}{\mathcal V}
\newcommand{\calE}{\mathcal E}
\newcommand{\calD}{\mathcal D}
\newcommand{\calL}{\mathcal L}

\newcommand{\calZ}{\mathcal Z}
\newcommand{\calP}{\mathcal P}

\newtheorem{thm}{Theorem}[section]
\newtheorem{lmm}[thm]{Lemma}
\newtheorem{prop}[thm]{Proposition}
\newtheorem{cor}[thm]{Corollary}

\theoremstyle{definition}              

\newtheorem{example}{Example}

\newtheorem{defn}[thm]{Definition}

\newcommand{\defeq}{\stackrel{\text{def}}{=}}

\newcommand{\R}{\mathbb R} 
\newcommand{\Rp}{{\mathbb R}_{\tiny +}} 

\newcommand{\cd}{\text{ :- }}

\newcommand{\bigjoin}{\mathlarger{\mathlarger{\mathlarger{\Join}}}}

\newcommand{\td}{\text{\sf TD}}
\newcommand{\inn}{\text{\sf in}}
\newcommand{\out}{\text{\sf out}}

\newcommand{\vars}{\text{\sf vars}}

\newcommand{\mon}{\text{\sf MON}}
\newcommand{\sub}{\text{\sf SUB}}

\newcommand{\panda}{\textsf{PANDA}\xspace}

\newcommand{\newalgorithm}{\textsf{PANDAExpress}\xspace}

\newcommand{\applystep}{\textsf{apply-step}\xspace}
\newcommand{\resetineq}{\textsf{reset}\xspace}
\newcommand{\emblue}[1]{{\color{blue} {#1}}}

\newcommand{\fhtw}{\textsf{fhtw}}
\newcommand{\subw}{\textsf{subw}}

\newcommand{\ov}{\overline}

\usepackage{fullpage}
\usepackage{amssymb}
\usepackage{mathrsfs}

\allowdisplaybreaks

\newcommand{\myemail}[1]{\small\texttt{#1}}

\begin{document}

\title{$\newalgorithm$: A Simpler and Faster PANDA Algorithm}
\author{
  \begin{tabular}{c}
    Mahmoud Abo Khamis\\
    \rai\\
    \myemail{mahmoud.abokhamis@relational.ai}
  \end{tabular}
  \quad
  \begin{tabular}{c}
    Hung Q. Ngo\\
    \rai\\
    \myemail{hung.ngo@relational.ai}
  \end{tabular}
  \quad
  \begin{tabular}{c}
    Dan Suciu\\
    University of Washington\\
    \myemail{suciu@cs.washington.edu}
  \end{tabular}
}

\begin{abstract}
    $\panda$ is a powerful generic algorithm for answering conjunctive queries
    and disjunctive datalog rules given input degree constraints. In the special
    case of Boolean queries with only cardinality constraints,
    $\panda$ runs in $\tilde O (N^{\subw})$-time, where $N$ is the
    input size, and $\subw$ is the submodular width of the query, a notion introduced by
    Daniel Marx (JACM 2013) in the context of constraint satisfaction problems. When
    specialized to certain classes of sub-graph pattern finding problems, the $\tilde
    O(N^{\subw})$ runtime matches the optimal runtime possible, modulo some
    conjectures in fine-grained complexity (Bringmann and Gorbachev (STOC 25)). The $\panda$
    framework is much more general, as it handles arbitrary input degree constraints,
    which capture common statistics and integrity constraints used in relational database
    management systems, it works for queries with free variables, and
    for both conjunctive queries and disjunctive datalog rules.

    The key weakness of $\panda$ is the large $\polylog(N)$-factor hidden in the $\tilde
    O(\cdot)$ notation, making it completely impractical, and falling short of what
    is achievable with specialized algorithms. This paper resolves this weakness with two
    novel ideas. First, we prove a new probabilistic inequality that upper-bounds the output
    size of disjunctive datalog rules under arbitrary degree constraints. Second, the proof
    of this inequality directly leads to a new algorithm named $\newalgorithm$ that is both
    {\em simpler} and {\em faster} than $\panda$. A novel feature of $\newalgorithm$ is a
    new partitioning scheme that uses arbitrary hyperplane cuts instead of axis-parallel
    hyperplanes used in $\panda$. These hyperplanes are dynamically constructed based on
    data-skewness statistics carefully tracked throughout the algorithm's execution. As a
    result, $\newalgorithm$ removes the $\polylog(N)$-factor from the runtime of $\panda$,
    matching the runtimes of intricate specialized algorithms, while retaining all its
    generality and power.
    As a bonus, we also show how $\newalgorithm$ can handle $\ell_p$-norm constraints,
    which generalize degree constraints.
\end{abstract}

\keywords{PANDA; conjunctive queries; disjunctive datalog rules; submodular width; proof sequences; Shannon inequalities}

\maketitle

\section{Introduction}
\label{sec:intro}

Conjunctive query (CQ) evaluation is a fundamental problem in many areas of computer
science, from traditional ones such as relational database management, graph
analytics, sparse tensor kernels, constraint satisfaction, logic
~\cite{DBLP:conf/pods/KhamisNR16, DBLP:conf/pods/GottlobGLS16, DBLP:journals/talg/GroheM14,
DBLP:books/aw/AbiteboulHV95}, to other settings such as deciding query
equivalence~\cite{DBLP:conf/stoc/ChandraM77} and equality
saturation~\cite{DBLP:journals/pacmpl/ZhangWWT22}.
There has been a lot of work in both
theory and practice to develop efficient query planning and evaluation techniques for
evaluating CQs for at least
50 years~\cite{DBLP:conf/sigmod/SelingerACLP79, DBLP:journals/ftdb/DingNC24}.

The past 15 years has witnessed an exciting new paradigm for optimizing and
evaluating CQs, involving two key ideas. First, we can derive a non-trivial worst-case
output size bound of a CQ, or more generally a disjunctive datalog rule, given
commonly collected input
statistics~\cite{MR859293,
MR2104047,
DBLP:journals/jacm/GottlobLVV12,
DBLP:conf/pods/KhamisNS16,
DBLP:journals/siamcomp/AtseriasGM13,
DBLP:conf/pods/Khamis0S17}.
Second, one can turn a {\em proof} of this
bound into a query plan to answer the query
efficiently~\cite{DBLP:journals/sigmod/NgoRR13,
DBLP:conf/pods/NgoPRR12,DBLP:conf/icdt/Veldhuizen14,
theoretics:13722}.
These new query plans are completely non-trivial and very different from the traditional
query plans based on select-project-join
trees~\cite{DBLP:conf/lics/Suciu23,10.1145/3196959.3196990,theoretics:13722}.

These new query plans boil down to three basic operations: {\em partitioning} the data based
on the frequencies of the values in some columns, {\em joining} two or more relations, and
{\em projection}. Instances of this type of query plans can be traced back to a special case
of CQ evaluation: the problem of finding a small (hyper)graph pattern inside a large
graph~\cite{DBLP:conf/stoc/Itai77,MR599482, MR1639767,DBLP:conf/soda/GroheM06,
DBLP:journals/talg/GroheM14,DBLP:journals/siamcomp/AtseriasGM13}. For example, Itai and
Rodeh~\cite{DBLP:conf/stoc/Itai77} described an $O(|E|^{3/2})$-time triangle detection
algorithm by partitioning vertices into heavy and light based on degree threshold
$|E|^{1/2}$. Alon, Yuster, and Zwick~\cite{DBLP:journals/algorithmica/AlonYZ97} generalized
the idea to the $k$-cycle detection problem with runtime $O(|E|^{2-\frac{1}{\lceil k/2
\rceil}})$ by adjusting the threshold and distributing the computational load to different
tree decompositions of the $k$-cycle graph. Recently, algorithms to find much more complex
graph patterns were discovered by Bringmann and
Gorbachev~\cite{DBLP:conf/stoc/BringmannG25,DBLP:journals/corr/abs-2404-04369}.
While the previous algorithms partition the data into two classes, the algorithm
in~\cite{DBLP:conf/stoc/BringmannG25} partitions the data into a number of classes that
depends on the graph pattern (the query), which is still $O(1)$ in the size of the input
database.

In an RDBMS, queries and data statistics are more general than those
in the graph-pattern finding problem. CQs have relations of arbitrary
arities, the output has arbitrary free variables, and the data statistics can include not
only relation sizes, but also ``degree constraints'' which bound the number of distinct
values in some columns given fixed values in other columns~\cite{DBLP:conf/pods/KhamisNS16}.

The pioneering work of Marx~\cite{DBLP:journals/jacm/Marx13} deals with Boolean CQs over
relations with arbitrary arities, using multiple tree decompositions to
load-balance. Let $\subw(Q)$ denote the {\em submodular width} of a Boolean query $Q$,
then Marx's algorithm runs in time $\tilde O(N^{c \cdot \subw(Q)})$ to answer $Q$ on any
database of size $N$, and $c>1$ is a constant. While Marx's algorithm does not handle
degree constraints and arbitrary free variables, his notion of submodular width is very
powerful. In addition to it being the optimal parameter for fixed-parameter  tractability of
CSP~\cite{DBLP:journals/jacm/Marx13}, it also is the precise {\em optimal} parameter for a
class of sub-graph identification and listing queries in the fined-grained complexity
setting, as shown very recently by Bringmann and
Gorbachev~\cite{DBLP:conf/stoc/BringmannG25}.

The natural question arises whether any CQ (Boolean or not) under arbitrary
degree constraints can be computed optimally by the same strategy of partitioning the data
into $O(1)$ classes, based on degrees, and performing joins. A general algorithm called
$\panda$ for computing CQs was described in~\cite{theoretics:13722}, and it
has the best known runtime for answering an arbitrary CQ using a
combinatorial algorithm (i.e. without using Fast Matrix Multiplication), except for an
undesired large polylogarithmic factor. This undesired factor is due to the fact that, at
each partition step, $\panda$ needs to partition a relation into $\log N$ parts, and  the
number of partitioning steps is a function of the query. In particular, $\panda$'s runtime
is $O\left(N^{\subw(Q)}\polylog(N) + |Q|\right)$, where $\subw(Q)$ is the submodular width
of the query $Q$, $N$ is the input size, and $|Q|$ is the output
size.

The polylog factor means $\panda$ provides an unsatisfactory answer to our desire of
understanding the complexity of query evaluation, despite its generality and power.
Its shortcoming does not stop at combinatorial algorithms, but also extends to
algorithms that use Fast Matrix Multiplication (FMM).  All prior algorithms using
FMM~\cite{DBLP:journals/algorithmica/AlonYZ97, DBLP:journals/siamcomp/DalirrooyfardVW21,
DBLP:conf/stoc/DalirrooyfardMW24, DBLP:journals/tcs/EisenbrandG04} apply the same principle
of partitioning the data into $O(1)$ classes, then using either joins or FMM in each class.
In contrast, the generalization of $\panda$ to using FMM described
in~\cite{DBLP:journals/pacmmod/KhamisHS25} also partitions the data into polylogarithmic
number of classes, falling short of the special cases described in the literature by
a polylogarithmic factor.

This unsatisfactory state of the art naturally leads to the question whether it is possible
to design an algorithm that shaves off $\panda$'s polylogarithmic factor while retaining all
its generality. On a closer examination, we found that the specialized algorithms for
graph-pattern finding in the
literature~\cite{DBLP:journals/algorithmica/AlonYZ97,DBLP:conf/stoc/BringmannG25} not only
partition the data into $O(1)$ parts at each step, but have another important common
characteristic: their partitioning criterion consists of comparing a single degree to a
pre-computed threshold (``heavy" or ``light").  We call this an \emph{axis-parallel}
partitioning method; because, as we shall explain in Section~\ref{sec:motivations}, the
strategy corresponds to partitioning the input relations using axis-parallel hyperplanes.
$\panda$ also does an axis-parallel partitioning; it just uses many more bins. The question
then becomes whether it is possible to design an algorithm that answers any CQ in submodular
width time using only $O(1)$ axis-parallel parts at each step.

{\bf Contributions.} This paper shows that the answer to the above question is likely to be
negative, and describes instead a very simple algorithm that answers a CQ in time
$O\left(N^{\subw(Q)}\log N  + |Q|\right)$ by using partitions based on hyperplane cuts that
are not necessarily axis-parallel.  (The remaining log-factor represents a sorting step.)

Specifically, the main problem we consider is the worst-case-optimally answering of
\emph{Disjunctive Datalog Rules} (DDR), which are {\em the} main building block for
answering CQs using the $\panda$
framework~\cite{DBLP:conf/pods/Khamis0S17,theoretics:13722}. DDRs capture the essence of
data partitioning: the body of a DDR is a conjunction, the head is a disjunction of atoms,
and the goal is to ``partition'' the outputs of the CQ among the head atoms so as to
minimize the largest cardinality among the output relations.

The measuring stick for the notion of worst-case optimality~\cite{10.1145/3196959.3196990}
is a tight bound on the output size of a DDR given input degree constraints. Our {\em first
contribution} is a new probabilistic inequality that can be used to prove the same tight
output size bounds for DDRs as the one shown
in~\cite{DBLP:conf/pods/Khamis0S17,theoretics:13722}.

Our {\em second contribution} is a new algorithm, called $\newalgorithm$, derived from the
proof of the probabilistic inequality. The algorithm answers DDRs in time $O((N+B)\log N)$
where $N$ is the input size and $B$ is the worst-case size bound of the DDR under the given
degree constraints. As shown in~\cite{DBLP:conf/pods/Khamis0S17,theoretics:13722}, this
result immediately implies an algorithm for answering CQs in time
$O\left(N^{\subw(Q)}\log N + |Q|\right)$, where $\subw(Q)$ is the submodular width of the
query $Q$, and $|Q|$ is the output size. As can be seen in Algorithm~\ref{alg:tight:panda},
$\newalgorithm$ is {\em breathtakingly simple}, which is another key takeaway
from this work: it is not only faster but also much simpler than $\panda$.

As a {\em third contribution}, we also show that the same algorithm can be used to handle
$\ell_p$-norm constraints~\cite{DBLP:journals/pacmmod/KhamisNOS24}, which generalize degree
constraints. These constraints have been shown to be very useful in practice for query
optimization~\cite{DBLP:journals/pacmmod/ZhangMKOS25}.

The crucial insights are as follows. First, we show that $O(1)$ axis-parallel partitions are
{\em insufficient} for optimality in the general case of DDRs. While this worked fine for
all special cases studied in the literature~\cite{DBLP:conf/stoc/Itai77,
DBLP:journals/algorithmica/AlonYZ97, DBLP:conf/stoc/BringmannG25,
DBLP:journals/corr/abs-2404-04369}, it fails in the general case. Instead, $\newalgorithm$
partitions the input and intermediate relations using $O(1)$ {\em arbitrary hyperplanes}.
Second, the hyperplanes separating the data are also not statically determined. The
algorithm collects and maintains data-skew statistics along the way, which are then used to
determine the hyperplanes to partition the data. The dynamic partitioning strategy is
designed to {\em load-balance} amongst sub-query plans in a fine grained manner. This level
of precision was not possible in query plans that make use of only polymatroid-based proof
sequences as in the original $\panda$ algorithm, because the polymatroids model the
entropies which are too coarse to capture data-level skewness.

{\bf Outline.} The rest of the paper is organized as follows. Section~\ref{sec:prelims}
briefly presents background knowledge needed for the paper. Section~\ref{sec:motivations}
presents examples and motivations for our work; in particular, it presents some examples
where axis-parallel partitioning alone is insufficient. Section~\ref{sec:prob-inequality}
states and proves our new probabilistic inequality. Section~\ref{sec:algorithm} presents the
new $\panda$ algorithm to answer DDRs, based on the new inequality.
Section~\ref{sec:correctness:runtime} proves the correctness of the algorithm and analyzes
its runtime. Section~\ref{sec:subw:ddr} explains how to use the algorithm for answering CQs
in submodular width time.
Section~\ref{sec:lp:norm} shows how to use the algorithm to handle $\ell_p$-norm
constraints.
We conclude in Section~\ref{sec:conclusion} with some final
remarks and future directions.

\section{Background}
\label{sec:prelims}

\subsection{Polymatroids and Shannon (Flow) Inequalities}

Let $\bm V$ be a set of variables. Let $\bm X, \bm Y \subseteq \bm V$ be disjoint sets of
variables, then the pair $\delta = (\bm Y | \bm X)$ is called a {\em monotonicity term}. A
monotonicity term of the form $\delta = (\bm Y | \emptyset)$ is called an {\em unconditional
monotonicity term}, in which case we simply write $\delta = \bm Y$. Let $\bm X, \bm Y, \bm Z
\subseteq \bm V$ be $3$ disjoint sets of variables, then the triple $\sigma = (\bm Y ; \bm Z
| \bm X)$ is called a {\em submodularity term}. Let $\mon$ and $\sub$ denote the sets of all
monotonicity terms and submodularity terms over $\bm V$, respectively.

For any two sets $\bm X, \bm Y$ of variables, we write $\bm X\bm Y$ for $\bm X \cup \bm Y$.
Given a function $h : 2^{\bm V} \to \R_+$, define
\begin{align}
    h(\delta) &\defeq h(\bm X \bm Y) - h(\bm X), && \delta = (\bm Y | \bm X) \\
    h(\sigma) &\defeq h(\bm X \bm Y) + h(\bm X  \bm Z) - h(\bm X  \bm Y  \bm Z) - h(\bm X)
    = h(\bm Y | \bm X) - h(\bm Y | \bm X \bm Z),
        && \sigma = (\bm Y ; \bm Z | \bm X)
\end{align}
The function $h$ can also be viewed interchangeably as a vector $\bm h = (h(\bm S))_{\bm S
\subseteq \bm V} \in \R_+^{2^{\bm V}}$. A function $h : 2^{\bm V} \to \R_+$ is called a {\em
polymatroid} if it satisfies all the {\em elemental Shannon inequalities}:
\begin{align}
    h(\emptyset) &= 0
    && h(\delta) \geq 0, \;\forall \delta \in \mon
    && h(\sigma) \geq 0, \;\forall \sigma \in \sub,
    \label{eqn:basic:shannon}
\end{align}
where $h(\delta) \geq 0$ is called the {\em monotonicity}
and $h(\sigma) \geq 0$ the {\em submodularity} property.

{\em Shannon inequalities} (beyond the elemental ones) are all inequalities obtained by
taking non-negative linear combinations of the elemental Shannon
inequalities~\eqref{eqn:basic:shannon}.
In other words, Shannon inequalities are linear inequalities that hold for all polymatroids.
We are interested in Shannon inequalities of a
special form, called {\em Shannon-flow inequalities} in \cite{DBLP:conf/pods/Khamis0S17}.

\begin{defn}
Let $\Sigma$ denote a set of {\em unconditional} monotonicity terms,
and $\Delta$ denote another set of monotonicity terms.
Let $\bm \lambda = (\lambda_{\bm Z})_{{\bm Z} \in \Sigma}$ and
$\bm w = (w_\delta)_{\delta \in \Delta}$ be two sets of non-negative
{\em rational} coefficients with $\norm{\bm \lambda}_1=1$.
A {\em Shannon-flow inequality}
is a Shannon inequality of the form
\begin{align}
    \sum_{\bm Z \in \Sigma} \lambda_{\bm Z} \cdot h({\bm Z})
    &\leq \sum_{\delta \in \Delta} w_\delta \cdot h(\delta)
    \label{eqn:shannon:flow:inequality}
\end{align}
\end{defn}

Below is an equivalent characterization of Shannon-flow inequalities:

\begin{prop}[\cite{theoretics:13722}]
Inequality~\eqref{eqn:shannon:flow:inequality}
holds for all polymatroids iff there are multisets $\calZ$, $\calD$, $\calM$, $\calS$
over base sets $\Sigma$, $\Delta$, $\mon$, $\sub$ respectively, such that
\begin{align}
    \sum_{\bm Z \in \calZ} h(\bm Z) = \sum_{\delta \in \calD} h(\delta)
    - \sum_{\mu \in \calM} h(\mu)
    - \sum_{\sigma\in \calS} h(\sigma)
    \label{eqn:bag:identity}
\end{align}
holds as an identity over symbolic variables $h(\bm X)$, $\emptyset \neq \bm X \subseteq \bm V$.
Furthermore, the sizes of the multisets are functions of $\bm \lambda, \bm w$, and $|\bm V|$.
\label{prop:shannon:identity}
\end{prop}

\begin{defn}
Given identity~\eqref{eqn:bag:identity}, we recover the corresponding Shannon-flow inequality~\eqref{eqn:shannon:flow:inequality}
by dropping the negative terms on the right-hand side of~\eqref{eqn:bag:identity},
\begin{align}
    \sum_{\bm Z \in \mathcal Z} h(\bm Z)
    &\leq
    \sum_{\delta \in \mathcal D} h(\delta)
    \label{eqn:shannon:flow:inequality:bag}
\end{align}
and dividing both sides by $|\calZ|$.
We refer to inequality~\eqref{eqn:shannon:flow:inequality:bag} an
{\em integral Shannon-flow inequality}, that is parameterized by the multisets
$(\calZ, \calD)$, and {\em witnessed} by the multisets $(\calM, \calS)$.
\end{defn}

The following two key results about (integral) Shannon-flow inequalities were shown
in~\cite{DBLP:conf/pods/Khamis0S17,theoretics:13722}.
They are both straightforward consequences of Proposition~\ref{prop:shannon:identity}.
See Appendix~\ref{app:supporting:results} for a self-contained proof of these results.

\begin{restatable}[The Proof Sequence Construction~\cite{theoretics:13722,DBLP:conf/pods/Khamis0S17}]{thm}{ThmProofSequence}
Every integral Shannon-flow inequality~\eqref{eqn:shannon:flow:inequality:bag} parameterized
by $(\calZ, \calD)$ with witness $(\calM, \calS)$ can be proved by a sequence of steps
\[
\sum_{\delta \in \calD} h(\delta) =
\sum_{\delta \in \calD_{0}} h(\delta)
\geq
\sum_{\delta \in \calD_{1}} h(\delta)
\geq
\cdots
\geq
\sum_{\delta \in \calD_{k}} h(\delta), \text{ with }
\calD_k \supseteq \calZ
\]
where
$\sum_{\delta \in \calD_{i}} h(\delta)$ is turned into
$\sum_{\delta \in \calD_{i+1}} h(\delta)$ by applying exactly one of the following steps:
\begin{itemize}
    \item {\em Submodularity:} $h(\bm Y|\bm X) \to h(\bm Y|\bm X\bm Z)$.
    \item {\em Composition:} $h(\bm X) + h(\bm Y|\bm X)\to h(\bm X \bm Y)$.
    \item {\em Decomposition:} $h(\bm X \bm Y) \to h(\bm X) + h(\bm Y|\bm X)$.
    \item {\em Monotonicity:} $h(\bm X \bm Y) \to h(\bm X)$.
\end{itemize}
Furthermore, $(\calZ, \calD_i)$ is an integral Shannon-flow inequality
witnessed by $(\calM_i, \calS_i)$, for all $0 \leq i \leq k$,
and the number of steps is bounded by $k \leq |\calD| + |\calM| + 3|\calS|$.
\label{thm:proof-sequence}
\end{restatable}
Later, Fig.~\ref{fig:eq:q:hex:proof} has an example of a proof sequence for a Shannon-flow inequality from Eq.~\eqref{eq:q:hex:inequality}.

\begin{restatable}[The Reset Lemma~\cite{theoretics:13722}]{lmm}{LmmReset}
Given an integral Shannon-flow inequality~\eqref{eqn:shannon:flow:inequality:bag}
parameterized by $(\calZ, \calD)$ and witnessed by $(\calM, \calS)$,
let $\bm W \in \calD$ be an
unconditional monotonicity term. Then, we can construct multisets $\calZ', \calD', \calM',
\calS'$ such that the following hold:
\begin{itemize}
    \item[(a)] $(\calZ', \calD')$ are parameters of another integral Shannon-flow inequality
witnessed by $(\calM', \calS')$,
    \item[(b)] $\calZ'\subseteq \calZ$ and $\calD'\subseteq \calD-\{\bm W\}$,
    with $|\calZ'| \geq |\calZ| - 1$.
    \item[(c)] $|\calD'| + |\calM'| + 2|\calS'| \leq |\calD| + |\calM| + 2|\calS| - 1$.
\end{itemize}
\label{lmm:reset-lemma}
\end{restatable}

In English, the Reset lemma allows us to drop one term $\bm W$ from the RHS
of Eq.~\eqref{eqn:shannon:flow:inequality:bag} and obtain a new inequality, by dropping at most one term from the LHS.  For a simple
illustration, consider the inequality $h(ABC)+h(BCD)\leq h(AB)+h(BC)+h(CD)$ (which we
prove in Section~\ref{subsec:motivation:axis:parallel}).  Suppose we want to drop $h(AB)$: to keep the inequality valid it suffices
to drop $h(ABC)$ from the LHS, and obtain $h(BCD)\leq h(BC)+h(CD)$.  Suppose instead we want
drop $h(BC)$: then we can drop either $h(ABC)$ or $h(BCD)$ from the LHS and the inequality
remains valid.

\subsection{Disjunctive datalog rules}

Let $\bm V$ be a set of variables. An {\em atom} is an expression of the form $R(\bm X)$,
where $R$ is a relation symbol, and $\bm X \subseteq \bm V$ is a set of variables.
A {\em schema} $\Sigma$ is a set of atoms.

Given a finite domain $\dom$ and variables $\bm X \subseteq \bm V$, let $\dom^{\bm X}$
denote the set of all tuples over the variables in $\bm X$ with values from $\dom$. Given a
tuple $\bm t \in \dom^{\bm V}$, we use $\bm t_{\bm X}$ to denote the projection of $\bm t$
onto the variables in $\bm X$. A {\em database instance} $D$ over schema $\Sigma$ is a
mapping that assigns each atom $R(\bm X) \in \Sigma$ to a finite relation $R^{D} \subseteq
\dom^{\bm X}$. Usually the database instance is clear from the context, so we simply write
$R$ instead of $R^D$.
Given a relation $R\subseteq \dom^{\bm X}$, we use $\vars(R)$ to denote the set of variables $\bm X$.

Given a database instance over schema $\Sigma$,
the {\em full natural join} of the instance is the set of tuples
$\bm t \in \dom^{\bm V}$ that satisfy all atoms in $\Sigma$:
\begin{align}
  \bigjoin \Sigma & \defeq \setof{\bm t \in \dom^{\bm V}}{
    \bm t_{\bm X} \in R,
    \text{ for all } R(\bm X) \in \Sigma
  }
  \label{eq:full:join}
\end{align}

\begin{defn}[DDR and its model]
Given two schemas $\Sigma_{\inn}$ and $\Sigma_{\out}$, a {\em disjunctive datalog rule} (DDR)
is the expression
\begin{align}
    \bigvee_{Q(\bm Z) \in \Sigma_{\out}} Q(\bm Z)
        &\cd  \bigwedge_{R(\bm X)\in \Sigma_{\inn}} R(\bm X)
    \label{eq:ddr}
\end{align}
Given a database instance over the input schema $\Sigma_{\inn}$,
a {\em model} (or {\em feasible output}) of the DDR~\eqref{eq:ddr}
is a database instance over the output schema $\Sigma_{\out}$ such that
for every tuple $\bm t \in \bigjoin \Sigma_{\inn}$, there exists at least one atom
$Q(\bm Z) \in \Sigma_{\out}$ such that $\bm t_{\bm Z} \in Q(\bm Z)$.
\label{defn:ddr}
\end{defn}

Without loss of generality, we assume that if $Q(\bm Z_1), Q(\bm Z_2) \in \Sigma_{\out}$
then $\bm Z_1 \neq \bm Z_2$. We will also write $\bm Z \in \Sigma_{\out}$ to mean
that $Q(\bm Z) \in \Sigma_{\out}$ for some relation symbol $Q$.
The {\em size} of an output instance to a DDR is the total number of tuples in all its relations:
\begin{align}
\norm{\Sigma_{\out}} \defeq \sum_{Q(\bm Z) \in \Sigma_{\out}} |Q|. \label{eqn:ddr:size}
\end{align}
A model of a DDR is {\em minimal} if no proper subset of it is also a model.
Note that a {\em conjunctive query} (CQ) is a special case of DDR where the output schema
$\Sigma_{\out}$ contains only one atom.
The answer to a CQ is its unique minimal model.

\subsection{Degree constraints and output size bounds for DDRs}

Consider a database instance over schema $\Sigma$.
Let $R \in \Sigma$ be a relation in this instance,
and $\delta = (\bm Y|\bm X)$ be a monotonicity term.
Let $\bm x \in \dom^{\bm X}$ be a data tuple.
Let $R_{|\bm X \cup \bm Y} \defeq \dom^{\bm X \cup \bm Y} \ltimes R$,\footnote{$S \ltimes T$ denotes
the {\em semi-join reduce} operator defined by $S \ltimes T \defeq \pi_{\vars(S)}(S \Join T)$.}
denote the {\em restriction} of $R$ onto $\bm X \cup \bm Y$;
this restriction is needed to define the degree notions below even in the case when
$\bm X \cup \bm Y \not\subseteq \vars(R)$.
Note that $R_{|\bm X \cup \bm Y} = \pi_{\bm X \cup Y} R$ when $\bm X \cup \bm Y \subseteq \vars(R)$.

Given a database instance over schema $\Sigma$, define the following three quantities:
\begin{align}
    \degree_{R}(\bm Y | \bm X = \bm x) &\defeq
    |\setof{\bm y \in \dom^{\bm Y}}{ (\bm x, \bm y) \in R_{|\bm X \cup \bm Y} }| \label{eq:degree:x} \\
    \degree_{R}(\delta) &\defeq \max_{\bm x \in \dom^{\bm X}} \degree_{R}(\bm Y|\bm X = \bm x) \label{eq:degree:delta} \\
  \degree_{\Sigma}(\delta) &\defeq \min_{R \in \Sigma}\degree_{R}(\delta).
  \label{eqn:degree:sigma}
\end{align}
Given a database instance over schema $\Sigma$,
if $N_\delta$ is an integer for which $\degree_{\Sigma}(\delta) \leq N_\delta$,
then we say that the schema $\Sigma$ satisfies the {\em degree constraint}
$(\delta, N_\delta)$.
More generally, let $\Delta\subseteq\mon$ be a set of monotonicity terms and
$\bm N = (N_\delta)_{\delta \in \Delta}$ be a vector of positive integers,
one for each monotonicity term in $\Delta$.
Then, $(\Delta, \bm N)$ is called a set of degree constraints, and we say
that the instance $\Sigma$ {\em satisfies}
$(\Delta, \bm N)$ if it satisfies every constraint in it, denoted by
$\Sigma \models (\Delta, \bm N)$.

The following upper bound from~\cite{theoretics:13722,DBLP:conf/pods/Khamis0S17}
is a generalization of the AGM bound~\cite{DBLP:conf/focs/AtseriasGM08}:

\begin{thm}[\cite{theoretics:13722,DBLP:conf/pods/Khamis0S17}]
\label{th:upper:bound}
Consider a DDR of the form~\eqref{eq:ddr} with input schema $\Sigma_{\inn}$, and output
schema $\Sigma_{\out}$ such that $\Sigma_{\inn} \models (\Delta, \bm N)$. Assume that there
exist two {\em non-negative} rational weight vectors $\bm w \defeq (w_\delta)_{\delta \in
\Delta}$ and $\bm \lambda \defeq (\lambda_{\bm Z})_{\bm Z \in\Sigma_{\out}}$ with $\norm{\bm
\lambda}_1=1$, where the following is a Shannon-flow inequality:
\begin{align}
    \sum_{\bm Z \in \Sigma_{\out}}\lambda_{\bm Z}\cdot h(\bm Z)
    & \leq \sum_{\delta \in \Delta} w_\delta \cdot h(\delta).
    \label{eqn:ddr:shearer}
\end{align}
Then, for any input instance $\Sigma_{\inn}$ of the DDR~\eqref{eq:ddr},
there exists a model $\Sigma_{\out}$ for the DDR for which:
  {
  \begin{align}
    \max_{\bm Z \in \Sigma_{\out}} |Q(\bm Z)| &\leq
      \prod_{\delta \in \Delta} N_\delta^{w_\delta}
      \label{eqn:ddr:output:size}
  \end{align}
  }
\end{thm}

\section{Motivating and Running Examples}
\label{sec:motivations}

This section presents two examples to illustrate axis-parallel partitioning and
the need for arbitrary hyperplane partitions. The
second example is also the running example for the rest of the paper.

\subsection{Axis-Parallel Partition}
\label{subsec:motivation:axis:parallel}

Consider the following DDR (Definition~\ref{defn:ddr}):
\begin{align}
  U(A,B,C) \vee V(B,C,D) &\cd R(A,B) \wedge S(B,C) \wedge T(C,D)\label{eq:q1}
\end{align}
The rule computes two target relations $U$, and $V$ such that, for any tuple $(a,b,c,d)$
satisfying the body, either $(a,b,c) \in U$ or $(b,c,d) \in V$.  In other words, we need to
``partition'' the output data into two parts, called $U$ and $V$.  The semantics is
non-deterministic, and the goal is to compute an output that minimizes $\max(|U|,|V|)$.

Assume all three input relations have size $\leq N$, we show that there exists an output
$(U, V)$ such that $|U|, |V| \leq N^{3/2}$. The worst-case optimal algorithm for answering
the DDR~\eqref{eq:q1} first partitions $S$ into the ``light'' and ``heavy'' parts, based on
how large the degree of a value $b$ with respect to $C$ is:
\begin{align}
    S^{\ell}(B,C) &:= \setof{(b,c)\in S}{\degree_S(C|B=b)\leq
        N^{1/2}},
 && S^{h}(B) := \setof{b}{\degree_S(C|B=b)> N^{1/2}}. \label{eqn:q1:joins:2}
\end{align}
Then, the algorithm computes $U$ and $V$ both of size $\leq N^{3/2}$ as follows:
\begin{align}
  U(A,B,C) &= R(A,B) \wedge S^{\ell}(B,C)
  && V(B,C,D) =  S^{h}(B) \wedge T(C,D)\label{eq:q1:joins}
\end{align}
The execution plan consisting of the steps \eqref{eqn:q1:joins:2}-\eqref{eq:q1:joins} is
derived from a logical query plan based on a proof of the following Shannon-flow inequality:
\begin{align}
  h(ABC)+h(BCD) &\leq h(AB)+h(BC)+h(CD) \label{eq:q1:inequality}
\end{align}
To prove this inequality, it suffices to observe that it can be written as a sum of two
simpler inequalities, which hold due to submodularity\footnote{By submodularity, $h(C|B)\geq
h(C|AB)$ and the first inequality follows from $h(AB)+h(C|B)\geq h(AB)+h(C|AB)=h(ABC)$,
while the second is already in the form of a submodularity inequality.}:
\begin{align}
  h(ABC) &\leq h(AB)+h(C|B) && h(BCD) \leq h(B)+h(CD) \label{eq:q1:chains}
\end{align}
Next, we describe the main intuition used to derive the execution
plan~\eqref{eqn:q1:joins:2}--~\eqref{eq:q1:joins} from the
inequalities~\eqref{eq:q1:chains}. Intuitively, assume that the polymatroid $h$ is the
entropy of some probability space over output tuples $(a,b,c,d)$ of the query.  In that
case, $h(AB) \leq \log|R|$ and, taking the base of the logarithm to be $N$ to
normalize the unit of measurement, we obtain $h(AB), h(BC), h(CD) \leq 1$. The RHS of
inequality~\eqref{eq:q1:inequality} is $h(AB)+h(BC)+h(CD)\leq 3$, which means that one of
the two inequalities~\eqref{eq:q1:chains} must have the RHS $\leq 3/2$.  This leads us to
the following axis-parallel partitioning of the space of polymatroids:
\begin{itemize}[leftmargin=*]
\item If $h(B)\geq 1/2$, then $h(C|B)=h(BC)-h(B)\leq 1/2$ and thus $h(ABC)\leq h(AB)+h(C|B)\leq 3/2$.
 In the space of degrees, this corresponds to $\degree_S(C|B) \leq
  N^{1/2}$, and we perform the first join in~\eqref{eq:q1:joins}.
\item If $h(B) < 1/2$, then $h(BCD) \leq h(B)+h(CD) < 3/2$.  This corresponds to
  $\degree_S(C|B) > N^{1/2}$ and we perform the second join in~\eqref{eq:q1:joins}.
\end{itemize}

The space of polymatroids $h$ is a subset of a 16-dimensional space $\Rp^{2^4}$, and the two
cases above partition this space using an axis-parallel hyperplane $h(B)=1/2$.

\subsection{The Hexagon Query and General Hyperplane Partitioning}
\label{subsec:motivation:hexagon}

Now, we present a different example where axis-parallel partitioning is insufficient.
Consider the following DDR, due to~\cite{yilei-wang-2021}:
\begin{align}
  Q(A,B,C,D,E,F) &\cd R(A,B,C) \wedge S(C,D,E) \wedge T(E,F,A) \wedge K(B,D,F)\label{eq:q:hex}
\end{align}
The DDR is just a standard CQ, which we call the \emph{hexagon query}. The query is not
challenging: if input relations have size $\leq N$, the query can be
computed using any worst-case optimal join (WCOJ)
algorithm~\cite{DBLP:conf/pods/NgoPRR12,DBLP:conf/icdt/Veldhuizen14,DBLP:journals/sigmod/NgoRR13},
in time $O\left(N^2\log N\right)$, which is optimal because the
AGM-bound\footnote{See the Appendix~\ref{app:examples} for more challenging examples.}
of the query is $N^2$.
However, these WCOJ
algorithms do not apply in general to DDRs, and we will not discuss them in this paper.
Instead, we want to compute the hexagon query using only data partitioning and joins, in a
manner similar to the previous example. The Shannon-flow inequality (Shearer's
lemma~\cite{MR859293}) that proves $|Q| \leq N^2$ is:
\begin{align}
  2h(ABCDEF) &\leq h(ABC)+h(CDE)+h(EFA)+h(BDF) \label{eq:q:hex:inequality}
\end{align}
This inequality is also a sum of two simpler inequalities, which form the basis of the proof
sequence in Theorem~\ref{thm:proof-sequence}:
\begin{align}
  h(ABCDEF) &\leq h(ABC)+h(DE|C) + h(F) &&
  h(ABCDEF) \leq h(EFA)+h(BD|F) + h(C)\label{eq:q:hex:chains}
\end{align}
Assuming that $h$
satisfies $h(ABC)$, $h(CDE)$, $h(EFA)$, $h(BDF) \leq 1$, we seek a partition of the space of
polymatroids such that, in each part, one of the two RHS
in~\eqref{eq:q:hex:chains} is $\leq 2$.  However, no axis-parallel partition exists with
this property.  If it existed, it would correspond to partitioning the values of $C$ and $F$
into light and heavy, but then the first join derived from~\eqref{eq:q:hex:chains} joins the
light-$C$ with heavy-$F$, while the second joins the heavy-$C$ with the light-$F$, which
leaves out the combinations light-light and heavy-heavy. Clearly, no algorithm is possible
by simply partitioning into light and heavy.

Instead, the natural partition of the polymatroids uses
the hyperplane $h(C)=h(F)$, as follows:
\begin{itemize}[leftmargin=*]
\item If $h(C)\geq h(F)$ then
  $h(ABCDEF) \leq h(ABC)+h(DE|C)+h(F) \leq h(ABC)+h(DE|C)+h(C)=
  h(ABC)+h(CDE)\leq 2$.
\item If $h(C) \leq h(F)$ then $h(ABCDEF) \leq h(EFA)+h(BD|F)+h(C)\leq
  h(EFA)+h(BD|F)+h(F)= h(EFA)+h(BDF) \leq 2$.
\end{itemize}
The corresponding hyperplane-based space partition is shown in
Fig.~\ref{fig:space-partition-a}. In order to compute the query by using only
axis-parallel partitions, $\panda$ uses a poly-logarithmic number of buckets:  It partitions the
values $c$ into $\log N$ buckets such that in bucket $i$, $2^{i-1}<
\degree_S(DE|C=c)\leq 2^i$, as shown in Fig.~\ref{fig:space-partition-b}, then partitions
similarly the values $f$ into $\log N$ buckets, as in Fig.~\ref{fig:space-partition-c}.
Then, in each of the $\log^2 N$ partitions, $\panda$ makes a choice whether to use the left or
the right inequality in~\eqref{eq:q:hex:chains}.  This is an expensive way to choose between
only two joins represented by the inequalities~\eqref{eq:q:hex:chains},
and this is the reason why $\panda$ has an extra $\log^2 N$ factor for this query.
In contrast, we will show later that our $\newalgorithm$ algorithm uses
only the hyperplane $h(C) = h(F)$ to partition this instance, thus removing the $\log^2 N$ overhead.

\begin{figure}
    \begin{subfigure}{0.3\textwidth}
        \begin{tikzpicture}[scale=.5, every node/.style={scale=.7}, >=stealth]
            \draw[->] (0, 0) -- (5.25,0) node[right] {$h(C)$};
            \draw[->] (0, 0) -- (0,5.25) node[above] {$h(F)$};
            \draw[thick] (0, 0) -- (5, 5);
        \end{tikzpicture}
        \caption{}
        \label{fig:space-partition-a}
    \end{subfigure}
    \begin{subfigure}{0.3\textwidth}
        \begin{tikzpicture}[scale=.5, every node/.style={scale=.7}, >=stealth]
            \draw[->] (0, 0) -- (5.25,0) node[right] {$h(C)$};
            \draw[->] (0, 0) -- (0,5.25) node[above] {$h(F)$};
            \foreach \x in {.5,1,1.5,2,2.5,3,3.5,4,4.5} {
                \draw (\x,0) -- (\x,5);
            }
        \end{tikzpicture}
        \caption{}
        \label{fig:space-partition-b}
    \end{subfigure}
    \begin{subfigure}{0.3\textwidth}
        \begin{tikzpicture}[scale=.5, every node/.style={scale=.7}, >=stealth]
            \draw[->] (0, 0) -- (5.25,0) node[right] {$h(C)$};
            \draw[->] (0, 0) -- (0,5.25) node[above] {$h(F)$};
            \foreach \x in {.5,1,1.5,2,2.5,3,3.5,4,4.5} {
                \draw (\x,0) -- (\x,5);
                \draw (0,\x) -- (5,\x);
            }
            \draw[thick] (0, 0) -- (5, 5);
        \end{tikzpicture}
        \caption{}
        \label{fig:space-partition-c}
    \end{subfigure}
    \caption{Space partition for the hexagon query in Eq.~\eqref{eq:q:hex}:
    While the {\em hyperplane partition} depicted in (a) is sufficient,
    $\panda$ uses $\log^2 N$ {\em axis-parallel} partitions, depicted in (b) and (c).}
    \label{fig:space-partition}
\end{figure}

\section{An Inequality on Sub-Probability Measures}
\label{sec:prob-inequality}

This section presents our first main technical contribution: a new inequality over {\em
sub-probability measures}, which is a concept that we review next. This inequality forms the basis of a new
algorithm $\newalgorithm$, and it is the probabilistic analogue of the Shannon-flow
inequality~\eqref{eqn:ddr:shearer} that underpins $\panda$.

\subsection{Sub-probability measures}
\label{subsec:sub:prob:measures}

We review some very basic facts from finite and discrete measure
theory~\cite{measure-theory-tao}. All our measures are finite and discrete in this paper, so
we will drop the adjectives ``finite and discrete" from now on.
Let $\Omega$ be a finite set. A function $p : 2^\Omega \to R_+$ is called a {\em measure}
on $\Omega$ if $p(\emptyset)=0$ and $p(A) = \sum_{x \in A} p(x)$
for any $A \subseteq \Omega$.
In particular, to define a measure on $\Omega$, it is sufficient to define the values
of $p(x)$ for all $x \in \Omega$.
It is called a {\em sub-probability measure} if $p(\Omega) \leq 1$,
and a {\em probability measure} if $p(\Omega) = 1$.

Let $p$ be a measure on $\Omega$. If $\Omega = \Omega_1 \times \Omega_2$ is a Cartesian
product of two domains, then the measure $p_1(x_1) \defeq \sum_{x_2 \in \Omega_2} p(x_1, x_2),
\forall x_1 \in \Omega_1$ is called the {\em marginal measure} of $p$ on $\Omega_1$. Let $A
\subseteq \Omega$ with positive measure $p(A) > 0$, then the measure $p_{|A}(B) \defeq \frac{p(B
\cap A)}{p(A)}$ for any $B \subseteq \Omega$ is called the {\em conditional measure} of $p$ on
$A$. Typically we write $p(B|A)$ instead of $p_{|A}(B)$. It is easy to see that, if $p$ is a
sub-probability measure then any of its marginals is also a sub-probability measure, and any
of its conditional measures is a probability measure.

In this paper, we will deal with domains $\Omega = \dom^{\bm X}$ where $\bm X$ is some
set of variables, and $\dom$ is a finite set. A sub-probability measure $p$ on $\dom^{\bm X}$
will be denoted by $p_{\bm X}$. The aforementioned facts are specialized as follows.

\begin{prop}
Given a sub-probability measure $p_{\bm X \bm Y}$, the following hold:
\begin{itemize}
\item The following marginal measure is a sub-probability measure:
\begin{align}
    p_{\bm X}(\bm x) \defeq \sum_{\bm y \in \dom^{\bm Y}} p_{\bm X \bm Y}(\bm x, \bm y)
\label{eqn:marginal}
\end{align}
\item Given $\bm x \in \dom^{\bm X}$
the following conditional measure is a sub-probability measure:
\begin{align}
    p_{\bm Y | \bm X = \bm x}(\bm y) &\defeq
        \begin{cases}
            \frac{p_{\bm X \bm Y}(\bm x, \bm y)}{p_{\bm X}(\bm x)} & \text{if } p_{\bm X}(\bm x) > 0 \\
            0 & \text{otherwise}
        \end{cases}
\label{eqn:conditional}
\end{align}
\end{itemize}
\label{prop:marginal:conditional}
\end{prop}
As is customary in probability theory, we write $p_{\bm Y | \bm X}(\bm y | \bm x)$ instead
of $p_{\bm Y | \bm X = \bm x}(\bm y)$, with the implicit understanding that there is a
measure for each $\bm x \in \dom^{\bm X}$.

\begin{defn}
Given $k$ sub-probability measures $p_1, p_2, \ldots, p_k$ on the same domain $\Omega$,
their {\em geometric mean} is a sub-probability measure $p$ defined by
$p(x) \defeq \left(\prod_{i=1}^k p_i(x)\right)^{\frac{1}{k}}, \forall x \in \Omega$.
\end{defn}

\begin{prop}
The geometric mean of sub-probability measures is also a sub-probability measure.
\label{prop:geometric:mean}
\end{prop}
\begin{proof}
This follows trivially from the AM-GM inequality:
\begin{align*}
    p(\Omega)
    = \sum_{x \in \Omega} \left(\prod_{i=1}^k p_i(x)\right)^{\frac{1}{k}}
    \leq \sum_{x \in \Omega} \frac{1}{k} \sum_{i=1}^k p_i(x)
    = \frac{1}{k} \sum_{i=1}^k \sum_{x \in \Omega} p_i(x)
    \leq 1
\end{align*}
\end{proof}

\subsection{The New Inequality}
\label{subsec:prob:inequality}

Given a monotonicity term $\delta = (\bm Y | \bm X)$,
a tuple $\bm t \in \dom^{\bm V}$ where $\bm V \supseteq \bm X \cup \bm Y$,
and a conditional sub-probability measure
$p_{\bm Y | \bm X}$, define the following shorthand notation:
\begin{align}
    p_{\delta}(\bm t) \defeq p_{\bm Y | \bm X}(\bm t_{\bm Y} | \bm t_{\bm X})
\end{align}
In particular, if $\delta = (\bm Z)$ was an unconditional monotonicity term,
then $p_{\bm Z}(\bm t) \defeq p_{\bm Z}(\bm t_{\bm Z})$.
We prove the following inequality.

\begin{lmm}
Let $\bm w \defeq (w_\delta)_{\delta \in \Delta}$ and $\bm \lambda \defeq (\lambda_{\bm Z})_{\bm Z
\in\Sigma_{\out}}$ with $\norm{\bm \lambda}_1=1$,
be any two non-negative rational weight vectors
for which~\eqref{eqn:ddr:shearer} is a Shannon-flow inequality over $\bm V$.
Then, given any collection of sub-probability measures $p_{\delta}$, $\delta \in \Delta$,
there exists a collection of sub-probability measures $p_{\bm Z}$, $\bm Z \in \Sigma_{\out}$,
such that
\begin{align}
    \prod_{Z \in \Sigma_{\out}} p_{Z}(\bm t)^{\lambda_Z}
    &\geq \prod_{\delta \in \Delta} p_{\delta}(\bm t)^{w_\delta},
    && \forall \bm t \in \dom^{\bm V}
\label{eqn:probabilistic:inequality}
\end{align}
\label{lmm:probabilistic:inequality}
\end{lmm}
\begin{proof}
Write $\lambda_{\bm Z} = \frac{m_{\bm Z}}{d}$ and $w_{\delta} = \frac{m_{\delta}}{d}$
where $m_{\bm Z}$ and $m_{\delta}$ are positive integers, and $d$ is the minimum common denominator
of the coefficients. By multiplying both sides of~\eqref{eqn:ddr:shearer} by $d$,
we obtain the integral Shannon-flow inequality~\eqref{eqn:shannon:flow:inequality:bag}:
$\sum_{\bm Z \in \calZ} h(\bm Z) \leq \sum_{\delta \in \calD} h(\delta)$
in which $\calZ$ is a multiset of elements in $\Sigma_{\out}$
and $\calD$ is a multiset of elements in $\Delta$.
Furthermore, elements $\bm Z \in \calZ$ appear with multiplicity $m_{\bm Z}$,
and elements $\delta \in \calD$ appear with multiplicity $m_{\delta}$.

We first claim that there exist sub-probability measures
$p_{\bm Z}$, one for each copy of $\bm Z$ in $\calZ$ for which
\begin{align}
    \prod_{\bm Z \in \calZ} p_{\bm Z}(\bm t)
    &\geq \prod_{\delta \in \calD} p_{\delta}(\bm t),
    && \forall \bm t \in \dom^{\bm V}
\label{eqn:probabilistic:inequality:bag}
\end{align}
To prove this claim, our strategy is to examine the proof sequence of the integral
Shannon-flow inequality from Theorem~\ref{thm:proof-sequence}, where each step turns the
multiset $\calD$ into a multiset $\calD'$ by applying one of the four proof steps. It is
sufficient to show that, for each proof step $\calD \to \calD'$, we can construct
sub-probability measures $p_{\delta}$ for all $\delta \in \calD'$ such that \(\prod_{\delta
\in \calD} p_{\delta}(\bm t) \leq \prod_{\delta \in \calD'} p_{\delta}(\bm t)\), because at
the last step we will obtain $\calD' \supseteq \calZ$, which implies $\prod_{\delta \in \calD'} p_{\delta}(\bm t) \leq \prod_{\bm Z \in \calZ} p_{\bm Z}(\bm t)$. This is done via trivial case analysis on the
four proof steps.

\begin{itemize}[leftmargin=*]
\item {\em Monotonicity:} $h(\bm X \bm Y) \to h(\bm X)$. Here, remove $\bm X\bm Y$ and add
$\bm X$ to $\calD$ to obtain $\calD'$. Correspondingly, we replace $p_{\bm X \bm Y}$ by
its marginal $p_{\bm X}$ defined in~\eqref{eqn:marginal}, where $p_{\bm X\bm Y}(\bm t) \leq
p_{\bm X}(\bm t)$ is trivial to verify.
\item {\em Submodularity:} $h(\bm Y | \bm X) \to h(\bm Y | \bm X \bm Z)$.
In this case, replace $p_{\bm Y | \bm X}$ by the conditional measure $p_{\bm Y | \bm X \bm Z}$,
defined by
$p_{\bm Y | \bm X \bm Z}(\bm y | \bm x\bm z) \defeq p_{\bm Y | \bm X}(\bm y | \bm x)$.
This way, we have not changed the product at all.
\item {\em Composition:} $h(\bm X) + h(\bm Y | \bm X) \to h(\bm X \bm Y)$.
Here, remove $\bm X$ and $\bm Y | \bm X$ from $\calD$ and add $\bm X \bm Y$ to
obtain $\calD'$.  Correspondingly, we keep the product unchanged by defining
$p_{\bm X \bm Y}(\bm x, \bm y) \defeq p_{\bm X}(\bm x) \cdot p_{\bm Y | \bm X}(\bm y | \bm x)$.
\item {\em Decomposition:} $h(\bm X \bm Y) \to h(\bm X) + h(\bm Y | \bm X)$.
Here, remove $\bm X \bm Y$ from $\calD$ and add $\bm X$ and $\bm Y | \bm X$ to
obtain $\calD'$.  Correspondingly, we keep the product unchanged by defining
the marginal and conditional measures exactly as shown in~\eqref{eqn:marginal}
and~\eqref{eqn:conditional}.
\end{itemize}
The claim is thus proved. Next, note that the RHS of~\eqref{eqn:probabilistic:inequality:bag}
is
\begin{align}
    \prod_{\delta \in \calD} p_{\delta}(\bm t)
    &= \prod_{\delta \in \Delta} p_{\delta}(\bm t)^{m(\delta)}
    = \prod_{\delta \in \Delta} p_{\delta}(\bm t)^{d \cdot w_\delta}
    = \left(\prod_{\delta \in \Delta} p_{\delta}(\bm t)^{w_\delta}\right)^d
\end{align}
Hence, it remains to show that there exists a set of sub-probability measures $p_{Z}$
for all $Z \in \Sigma_{\out}$ for which
\begin{align}
    \prod_{Z \in \Sigma_{\out}} p_{Z}(\bm t)^{m(Z)}
    &= \prod_{\bm Z \in \calZ} p_{\bm Z}(\bm t),
    && \forall \bm t \in \dom^{\bm V}
    \label{eqn:to:show}
\end{align}
To see this, for each $\bm Z \in \Sigma_{\out}$, set
$p_{\bm Z}$ to be the geometric mean of $m_{\bm Z}$ sub-probability
measures $p_{\bm Z}$ for copies of $\bm Z$ in $\calZ$,
then we obtain equality~\eqref{eqn:to:show}. The proof is complete because
by Proposition~\ref{prop:geometric:mean} the geometric mean of sub-probability measures
is also a sub-probability measure.
\end{proof}

Lemma~\ref{lmm:probabilistic:inequality} gives us another way to prove the output size bound
of DDR in Theorem~\ref{th:upper:bound}, and this proof will directly lead to the
$\newalgorithm$ algorithm described in the next section.

\begin{cor} Theorem~\ref{th:upper:bound} holds.
    \label{cor:upper:bound}
\end{cor}
\begin{proof}
For each $\delta  = (\bm Y | \bm X) \in \Delta$, let $R \in \Sigma_{\inn}$ be a relation
where $\max_{\bm x} \deg_R(\bm Y | \bm X = \bm x) \leq N_\delta$.
Define the following conditional sub-probability measure:
\begin{align}
    p_{\delta}(\bm y | \bm x) \defeq
    \begin{cases}
        \frac{1}{N_\delta} &
            (\bm x, \bm y) \in R_{|\bm X\cup\bm Y} \\
        0 & \text{otherwise}
    \end{cases}
\end{align}
Write $B \defeq \prod_{\delta \in \Delta} N_\delta^{w_\delta}$.
From Lemma~\ref{lmm:probabilistic:inequality}, there exists a sub-probability measure
$p_{\bm Z}$ for each $\bm Z \in \Sigma_{\out}$ such that
\begin{align}
    \prod_{\bm Z \in \Sigma_{\out}} p_{\bm Z}(\bm t)^{\lambda_{\bm Z}}
    &\geq \prod_{\delta \in \Delta} p_{\delta}(\bm t)^{w_\delta}
    = \prod_{\delta \in \Delta} \left(\frac{1}{N_\delta}\right)^{w_\delta}
    = \frac{1}{B},
    && \forall \bm t \in \bigjoin \Sigma_{\inn}
    \label{eqn:prob:bound}
\end{align}
To show~\eqref{eqn:ddr:output:size}, it is sufficient to show that there exists a model
$\Sigma_{\out}$ for the DDR~\eqref{eq:ddr} such that $|Q(\bm Z)| \leq B$ for all $Q(\bm Z)
\in \Sigma_{\out}$. This model is constructed by defining
\begin{align}
    Q(\bm Z) \defeq \left\{ \bm z \in \dom^{\bm Z} \ | \ p_{\bm Z}(\bm z) \geq \frac{1}{B} \right\}
    && \forall\bm Z \in \Sigma_{\out}
\end{align}
Obviously, $|Q(\bm Z)| \leq B$ because $p_{\bm Z}$ is a sub-probability measure. To show
that this is indeed a model for the DDR~\eqref{eq:ddr}, note that since $\norm{\bm
\lambda}_1 = 1$, for every $\bm t \in \bigjoin \Sigma_{\inn}$ inequality~\eqref{eqn:prob:bound}
implies that there exists some $\bm Z \in \Sigma_{\out}$ such that $p_{\bm Z}(\bm t) \geq
\frac{1}{B}$, which means $\bm t_{\bm Z} \in Q(\bm Z)$.
\end{proof}

\subsection{Illustration of the Inequality on the Hexagon Query}
\label{subsec:motivations:hexagon}

To illustrate why inequality~\eqref{eqn:probabilistic:inequality} holds and how it can be used to prove the
output size bound of DDRs, this section presents an alternative proof of the size bound $|Q|
\leq N^2$ for the hexagon query~\eqref{eq:q:hex}. Recall that the typical proof of this fact
is to make use of the AGM bound which is the direct consequence of Shearer's
lemma~\cite{MR859293}. For this query, Shearer's lemma is {\em exactly} the Shannon-flow
inequality~\eqref{eq:q:hex:inequality}.

We start from the inequality~\eqref{eq:q:hex:inequality} and write down the {\em proof
sequence} (Theorem~\ref{thm:proof-sequence}) that transforms the RHS into the LHS. This
proof sequence is depicted in Fig.~\ref{fig:eq:q:hex:proof} (left). Now, to each term on the RHS of~\eqref{eq:q:hex:inequality}, we
associate a probability measure. In particular, to $h(ABC)$, we associate a measure
$p_{ABC}$ defined as $p_{ABC}(abc) \defeq 1/N$ if the tuple $(a, b, c)$ is in $R$, and $0$
otherwise.  Similarly, we define measures $p_{CDE}$, $p_{EFA}$, and $p_{BDF}$ for $S$, $T$,
and $K$ respectively. By construction, we have:
\begin{align}
    p_{ABC}(abc) \cdot p_{CDE}(cde) \cdot p_{EFA}(efa) \cdot p_{BDF}(bdf) = 1/ N^4,
    \quad\quad\forall (a,b,c,d,e,f) \in Q.\label{eqn:product}
\end{align}
Then, we trace the proof steps in Fig.~\ref{fig:eq:q:hex:proof}: Each proof step replaces
(one or two) entropy terms by new entropy terms. We associate to each new entropy term a new
probability measure, thus mirroring the proof step in the measure world.
\begin{itemize}[leftmargin=*]
 \item From the first step $h(CDE)\to h(C) + h(DE|C)$ we
take the measure $p_{CDE}$ (associated to $h(CDE)$) and define two new measures $p_C$ and
$p_{DE|C}$ as its {\em marginal} and {\em conditional} (Eq.~\eqref{eqn:marginal}
and~\eqref{eqn:conditional}), and associate them to $h(C)$ and $h(DE|C)$ respectively. Note
that by definition, $p_{CDE} = p_C \cdot p_{DE|C}$, hence if we replace $p_{CDE}(cde)$ with
$p_C(c)\cdot p_{DE|C}(de|c)$ in Eq.~\eqref{eqn:product}, the equality still holds.
 \item Similarly, the second step $h(ABC) + h(DE|C) \to h(ABCDE)$ tells us to take $p_{ABC}$
and $p_{DE|C}$, define measure $p_{ABCDE}$ as their product, and associate it to $h(ABCDE)$.
To track this move, replace the product $p_{ABC}(abc) \cdot p_{DE|C}(de|c)$  in the new
Eq.~\eqref{eqn:product} with $p_{ABCDE}(abcde)$.
\item Following this process until the end of the proof sequence, we end up
with two different measures $p'(ABCDEF), p''(ABCDEF)$ corresponding to $2h(ABCDEF)$ on the
LHS of Eq.~\eqref{eq:q:hex:inequality}. At the same time, Eq.~\eqref{eqn:product} becomes:
\begin{align}
    p'_{ABCDEF}(abcdef) \cdot p''_{ABCDEF}(abcdef) = 1/N^4, \quad\quad\forall (a,b,c,d,e,f) \in Q.\label{eqn:product:final}
\end{align}
\item Finally, define a new measure $p_{ABCDEF}$ as the {\em geometric mean} of $p'_{ABCDEF}$ and
$p''_{ABCDEF}$ (Proposition~\ref{prop:geometric:mean}). Hence, $p_{ABCDEF}$ takes a value of
$1/N^2$ for all tuples $(a,b,c,d,e,f) \in Q$. This implies that the size of the support of
$p_{ABCDEF}$, which is exactly $|Q|$, is at most $N^2$, as desired.
\end{itemize}

\begin{figure}[ht]
{\small
\begin{tabular}{c | c}
    Proof Steps & Sub-probability Measures \\
    \hline
    $h(CDE)\to h(C) + h(DE|C)$&
        \begin{minipage}{0.65\textwidth}
            \vspace{.2cm}
            $\begin{aligned}
                p_{C}(c) &\defeq \sum_{de} p_{CDE}(cde) = \frac{\deg_S(DE | C = c)}{N}\\
                p_{DE|C}(de|c) &\defeq \frac{p_{CDE}(cde)}{p_C(c)} = \frac{1}{\deg_S(DE | C = c)}
            \end{aligned}$
            \vspace{.2cm}
        \end{minipage}
        \\\hline
    $h(ABC) + h(DE|C) \to h(ABCDE)$ &
        \begin{minipage}{0.65\textwidth}
            \vspace{.2cm}
            $\begin{aligned}
                p_{ABCDE}(abcde) &\defeq p_{ABC}(abc) \cdot p_{DE|C}(de|c) = \frac{1}{N\cdot \deg_S(DE | C = c)}
            \end{aligned}$
            \vspace{.2cm}
        \end{minipage}
        \\\hline
    $h(BDF) \to h(F) + h(BD|F)$ &
        \begin{minipage}{0.65\textwidth}
            \vspace{.2cm}
            $\begin{aligned}
                p_{F}(f) &\defeq \sum_{bd} p_{BDF}(bdf) = \frac{\deg_K(BD | F = f)}{N}\\
                p_{BD|F}(bd|f) &\defeq \frac{p_{BDF}(bdf)}{p_F(f)} = \frac{1}{\deg_K(BD | F = f)}
            \end{aligned}$
            \vspace{.2cm}
        \end{minipage}
        \\\hline
    $h(EFA) + h(BD|F) \to h(ABDEF)$ &
        \begin{minipage}{0.65\textwidth}
            \vspace{.2cm}
            $\begin{aligned}
                p_{ABDEF}(abdef) &\defeq p_{EFA}(efa) \cdot p_{BD|F}(bd|f) = \frac{1}{N\cdot \deg_K(BD | F = f)}
            \end{aligned}$
            \vspace{.2cm}
        \end{minipage}
        \\\hline
    $h(ABCDE) + h(F) \to h(ABCDEF)$ &
        \begin{minipage}{0.65\textwidth}
            \vspace{.2cm}
            $\begin{aligned}
                p'_{ABCDEF}(abcdef) &\defeq p_{ABCDE}(abcde) \cdot p_F(f) = \frac{\deg_K(BD | F = f)}{N^2 \cdot \deg_S(DE | C = c)}
            \end{aligned}$
            \vspace{.2cm}
        \end{minipage}
        \\\hline
    $h(ABDEF) + h(C) \to h(ABCDEF)$ &
        \begin{minipage}{0.65\textwidth}
            \vspace{.2cm}
            $\begin{aligned}
                p''_{ABCDEF}(abcdef) &\defeq p_{ABDEF}(abdef) \cdot p_C(c) = \frac{\deg_S(DE | C = c)}{N^2 \cdot \deg_K(BD | F = f)}
            \end{aligned}$
            \vspace{.2cm}
        \end{minipage}
\end{tabular}
}
\caption{A new proof of the bound $|Q| \leq N^2$ for the hexagon query $Q$ in
  Eq.~\eqref{eq:q:hex}, using sub-probability measures.
  The depicted proof
sequence is for the Shannon-flow inequality~\eqref{eq:q:hex:inequality} and it is a slightly more compacted version of
the proof sequence obtained from Theorem~\ref{thm:proof-sequence}, where we merge, for example, two
steps $h(DE|C)\to h(DE|ABC)$ and $h(ABC) + h(DE|ABC)\to h(ABCDE)$ into one $h(ABC) +
h(DE|C)\to h(ABCDE)$.
  For simplicity, we assume above that all divisions by zero produce zero (similar to Eq.~\eqref{eqn:conditional}).}
  \label{fig:eq:q:hex:proof}
\end{figure}

\section{The $\newalgorithm$ algorithm}
\label{sec:algorithm}

\subsection{Illustration of $\newalgorithm$ on the Hexagon Query}
\label{subsec:motivations:new-algorithm}

This section explains intuitively how the size bound proof of the Hexagon query shown in
Section~\ref{subsec:motivations:hexagon} can be turned into an actual algorithm. The
formal presentation of the $\newalgorithm$ algorithm is given in the next section.

The proof shown in Section~\ref{subsec:motivations:hexagon} can {\em almost} be turned
directly into an $O(N^2)$-algorithm for computing the hexagon query. All we need to do is
to actually {\em materialize} all the intermediate sub-probability measures defined in the
proof. Of course, we only need to materialize the non-zero values of each measure.

However, there is an issue we have to deal with: The measures $p'_{ABCDEF}$ and
$p''_{ABCDEF}$ defined before could potentially have supports larger than $N^2$, exceeding
the runtime budget of $O(N^2)$. The minor twist in our algorithm is to observe that we do
not need to compute the full measures $p'_{ABCDEF}$ and $p''_{ABCDEF}$. Instead of computing
them fully, we will only compute {\em their truncations}, $\ov p'_{ABCDEF}$ and $\ov
p''_{ABCDEF}$ respectively, where $\ov p'_{ABCDEF}$ results from $p'_{ABCDEF}$ by only
keeping the values that are $\geq 1/N^2$, and similarly for $\ov p''_{ABCDEF}$. By
construction, the supports of $\ov p'_{ABCDEF}$ and $\ov p''_{ABCDEF}$ are no larger than
$N^2$. Moreover, by Eq.~\eqref{eqn:product:final}, for every tuple $(a,b,c,d,e,f) \in Q$,
either $p'_{ABCDEF}(abcdef) \geq 1/N^2$ or $p''_{ABCDEF}(abcdef) \geq 1/N^2$ (or both),
hence the tuple $(a,b,c,d,e,f)$ will occur in the support of at least one of $\ov
p'_{ABCDEF}$ or $\ov p''_{ABCDEF}$. After computing $\ov p'_{ABCDEF}$ and $\ov
p''_{ABCDEF}$, we return the union of their supports as the output of the query.

So where does the partitioning hyperplane $h(C)=h(F)$ mentioned in Section~\ref{subsec:motivation:hexagon} come into play in this algorithm?
Consider the definitions of $p'_{ABCDEF}(abcdef)$ and $p''_{ABCDEF}(abcdef)$ in
Fig.~\ref{fig:eq:q:hex:proof}. Note that the condition $p'_{ABCDEF}(abcdef) \geq 1/N^2$
corresponds to the following inequality:
\begin{align}
    \deg_K(BD | F = f) &\geq \deg_S(DE | C = c),\label{eq:deg:inequality}
\end{align}
Similarly, the condition $p''_{ABCDEF}(abcdef) \geq 1/N^2$ corresponds to the opposite
inequality. Hence, in the degree space, the above algorithm is partitioning the output into
two parts, using the boundary $\deg_K(BD | F = f) = \deg_S(DE | C = c)$. This condition {\em
is} the interpretation of the hyperplane $h(C)=h(F)$ in the degree space, thus completing
the connection.

In summary, proof sequences form the basis of a {\em logical
query plan}, whereas the sub-probability measures form the basis for the actual {\em execution plan}
to answer the query.

\subsection{The Algorithm}

\begin{algorithm}[ht!]
    \caption{$\newalgorithm(\calZ, \calD, \calP)$}
    \label{alg:tight:panda}
    \begin{algorithmic}[1]
        \Statex {\bf Inputs:} $(\calZ, \calD)$ form an integral Shannon-flow inequality
        Eq.~\eqref{eqn:shannon:flow:inequality:bag}
        \Statex {\color{white}\bf Inputs:} $\calP$ is a set of sub-probability measures, one for each $\delta \in \calD$
        \vspace{-.25cm}\Statex\hrulefill
        \If {$\calD$ contains any $\bm Z \in \calZ$}
            \State \Return $Q(\bm Z) := \{ \bm z \ | \ p_{\bm Z}(\bm z) > 0 \}$
            \algorithmiccomment{Tuples in the support of $p_{\bm Z}$}
        \EndIf
        \State Let $s$ be the first step in a proof sequence of the
        Shannon-flow inequality defined by $(\calZ, \calD)$.
        \State $(\calD^\ell, \calP^\ell) \gets \applystep(s, \calD, \calP)$
        \State $\Sigma_{\out}^{\ell} \gets \newalgorithm(\calZ, \calD^\ell, \calP^\ell)$
        \algorithmiccomment{Light branch; continue on with the proof sequence}
        \If {$s$ is not a composition step {\bf or } $|\calZ| = 1$}
            \State \Return $\Sigma_{\out}^{\ell}$
            \algorithmiccomment{No heavy branch}
        \ElsIf {$s$ is a composition step $h(\bm X) + h(\bm Y | \bm X) \to h(\bm X \bm Y)$
            {\bf and} $|\calZ| > 1$}
        \label{alg:tight:panda:reset}
            \State $(\calZ^h, \calD^h) \gets \resetineq(\calZ, \calD^\ell, \bm X\bm Y)$
            \algorithmiccomment{Apply Reset Lemma~\ref{lmm:reset-lemma} with $\bm W = \bm X\bm Y$}
            \State $\calP^h \gets \{ p_{\delta} \ | \ p_{\delta}\in \calP^\ell \text{ and } \delta \in \calD^h \}$
            \algorithmiccomment{Note that $\calD^h \subseteq \calD^\ell - \{\bm X\bm Y\}$}
            \State $\Sigma_{\out}^{h} \gets \newalgorithm(\calZ^h, \calD^h, \calP^h)$
            \algorithmiccomment{Heavy branch}
            \State \Return $\Sigma_{\out}^{\ell} \cup \Sigma_{\out}^{h}$
        \EndIf
    \end{algorithmic}
\end{algorithm}

We are now ready to formally present our $\newalgorithm$ algorithm, given in
Algorithm~\ref{alg:tight:panda}. We shall prove its
correctness and analyze its runtime in the next section. The algorithm takes as input three
parameters: $\calZ$ is a multiset of unconditional monotonicity terms, $\calD$ is a multiset
of monotonicity terms, and $\calP$ is a set of sub-probability measures, one for each copy
of monotonicity term in $\calD$. The pair $(\calZ, \calD)$ are parameters of an integral
Shannon-flow inequality of the form~\eqref{eqn:shannon:flow:inequality:bag} witnessed by
$(\calM, \calS)$.
We think of $\calP$ as an {\em annotated} database instance where each tuple is annotated with a probability. We explain later (Eq.~\eqref{eq:initial:measure}) how to initialize these annotations from the input database instance to a DDR.
There is a global ``budget'' parameter $B$ which is a positive integer
that we will define later (Eq.~\eqref{eqn:budget:B}) in the proof of correctness.

The algorithm is recursive. In the base case, if $\calD$ contains any $\bm Z \in \calZ$,
then we can directly construct and return a relation $Q(\bm Z)$ by collecting tuples in the
support of the corresponding sub-probability measure $p_{\bm Z}$.

If the base case does not hold, then from Theorem~\ref{thm:proof-sequence} there exists a
proof sequence for Eq.~\eqref{eqn:shannon:flow:inequality:bag} whose first step is $s$. The
sub-routine $\applystep$ modifies $\calD$ and $\calP$ slightly to obtain a new multiset of
monotonicity terms $\calD^\ell$ and a new set of sub-probability measures $\calP^\ell$. We will
explain shortly what this modification is.
It then recursively calls the algorithm with the new
parameters $(\calZ, \calD^\ell, \calP^\ell)$, applying the next step in the proof sequence.
This is the ``light'' branch of the recursion.
The result is returned if the step is {\em not} a composition step, in which case there is no
``heavy'' branch.

If $s$ is a composition step and $|\calZ| > 1$, then we need to do more work, creating a heavy branch.
The $\resetineq$ call applies Lemma~\ref{lmm:reset-lemma} with $\bm W = \bm
X \bm Y$ to obtain a new Shannon-flow inequality defined by $(\calZ^h, \calD^h)$, where
$\calZ^h \subseteq \calZ$, $|\calZ^h| \geq |\calZ|-1$, and $\calD^h \subseteq \calD^\ell - \{\bm
X\bm Y\}$. We then define $\calP^h$ by restricting $\calP^\ell$ to only those $p_{\delta}$ where
$\delta \in \calD^h$. With these new parameters, we recursively call the algorithm again to
obtain $\Sigma_{\out}^{h}$. Finally, the union of the heavy and light branches
is returned.

To complete the description of the algorithm, we now explain the sub-routine $\applystep$,
which modifies $\calD$ and $\calP$ according to the proof step $s$.
\begin{itemize}[leftmargin=*]
    \item If $s$ is a \emph{decomposition} step $h(\bm X \bm Y) \to h(\bm X) + h(\bm Y |
    \bm X)$, then let $p_{\bm X}$ be the marginal measure~\eqref{eqn:marginal} of $p_{\bm X
    \bm Y}$ on $\bm X$ and $p_{\bm Y | \bm X}$ be the conditional
    measure~\eqref{eqn:conditional} of $p_{\bm X \bm Y}$ on $\bm Y$ given $\bm X$.
    \begin{align*}
        \calD^\ell &\gets \calD \setminus \{\bm X \bm Y\} \cup \{\bm X, (\bm Y | \bm X)\} &&
        \calP^\ell \gets \calP \setminus \{ p_{\bm X \bm Y} \} \cup \{ p_{\bm X}, p_{\bm Y | \bm X} \}
    \end{align*}
    \item If $s$ is a {\em submodularity} step $h(\bm Y | \bm X) \to h(\bm Y | \bm X \bm Z)$,
    then define $p_{\bm Y | \bm X \bm Z} \defeq p_{\bm Y | \bm X}$ and
    \begin{align*}
        \calD^\ell &\gets \calD \setminus \{ (\bm Y | \bm X) \} \cup \{ (\bm Y | \bm X \bm Z) \} &&
        \calP^\ell \gets \calP \setminus \{ p_{\bm Y | \bm X} \} \cup \{ p_{\bm Y | \bm X \bm Z} \}
    \end{align*}
    \item If $s$ is a {\em monotonicity} step $h(\bm X \bm Y) \to h(\bm X)$, then define
    $p_{\bm X}$ as the marginal measure~\eqref{eqn:marginal} of $p_{\bm X \bm Y}$ on $\bm X$,
    and
    \begin{align*}
        \calD^\ell &\gets \calD \setminus \{ \bm X \bm Y \} \cup \{ \bm X \} &&
        \calP^\ell \gets \calP \setminus \{ p_{\bm X \bm Y} \} \cup \{ p_{\bm X} \}
    \end{align*}
    \item If $s$ is a {\em composition} step $h(\bm X) + h(\bm Y | \bm X) \to h(\bm X \bm Y)$,
    then
    \begin{align*}
        \calD^\ell &\gets \calD \setminus \{ \bm X, (\bm Y | \bm X) \} \cup \{ \bm X \bm Y \} &&
        \calP^\ell \gets \calP \setminus \{ p_{\bm X}, p_{\bm Y | \bm X} \} \cup \{ p_{\bm X \bm Y} \}
    \end{align*}
    where $p_{\bm X\bm Y}$ is the {\em truncated} product measure:
    \begin{align}
        p_{\bm X \bm Y}(\bm x, \bm y) \defeq
        \begin{cases}
            p_{\bm X}(\bm x) \cdot p_{\bm Y | \bm X}(\bm y | \bm x),
            & \text{if }
            p_{\bm X}(\bm x) \cdot p_{\bm Y | \bm X}(\bm y | \bm x) \geq 1/B \\
            0, & \text{otherwise}
        \end{cases}
        \label{eqn:truncated:product}
    \end{align}
\end{itemize}

\section{Proof of Correctness and Runtime Analysis}
\label{sec:correctness:runtime}

Theorem~\ref{thm:tight:panda} is the main result of the paper. Before stating and proving
it, we need to introduce an auxiliary optimization problem and its properties, which are of
independent interest and will be used again in Section~\ref{sec:subw:ddr}.

\subsection{An Auxiliary Optimization Problem}

In order to minimize the upper bound of the output size, as given in
Theorem~\ref{th:upper:bound}, which is also the same runtime expression in
Theorem~\ref{thm:tight:panda}, we want to find coefficients $\bm w, \bm \lambda$ that
minimize the right-hand side of~\eqref{eqn:ddr:output:size} subject to the Shannon-flow
inequality~\eqref{eqn:ddr:shearer}. In particular, we want to solve the optimization problem
stated in Proposition~\ref{prop:opt:ddr} below. The result was implicit
in~\cite{theoretics:13722,DBLP:conf/pods/Khamis0S17}; however, we include the proof here for
completeness, in part because the language of the statement in the proposition is about the
primal problem, while the one in~\cite{theoretics:13722,DBLP:conf/pods/Khamis0S17} is about
the dual problem, and hence they look quite different.

\begin{defn}
    Let $(\Delta, \bm N)$ denote a set of degree constraints over variables $\bm V$.
    We write $h \models (\Delta, \bm N)$ to denote the fact that $h$ is a polymatroid
    on $\bm V$ that satisfies all degree constraints in $(\Delta, \bm N)$, i.e.,
    for every $\delta = (\bm Y|\bm X) \in \Delta$, we have
    $h(\delta) \defeq h(\bm X \bm Y) - h(\bm X) \leq \log N_\delta$.
\end{defn}
Note that, including the constraints that $h$ is a polymatroid, the set of constraints
$h \models (\Delta, \bm N)$ can be expressed as a set of linear inequalities.

The following lemma shows that the optimization problem in Proposition~\ref{prop:opt:ddr} below
can be converted into a linear program. We isolate this lemma because we will use it
again to convert the primal form of the fractional hypertree width and
the submodular width into their dual forms in Section~\ref{sec:subw:ddr}.
Note that the lemma was basically proved in Lemma 6.4 in~\cite{theoretics:13722}, but it
was stated using the dual form of the optimization problem.

\begin{lmm}[~\cite{theoretics:13722}]
\label{lmm:ddr:lp}
Let $(\Delta,\bm N)$ be a set of degree constraints over variables $\bm V$.
Then, the following optimization problem
\begin{align*}
    \min_{\bm \lambda, \bm w} \quad &
    \sum_{\delta \in \Delta} w_\delta \log N_\delta \\
    \text{s.t.} \quad &
    \sum_{\bm Z \in \Sigma_{\out}}\lambda_{\bm Z}\cdot h(\bm Z)
    \leq \sum_{\delta \in \Delta} w_\delta \cdot h(\delta)
    \text{ is a Shannon-flow inequality}\\
    & \norm{\bm \lambda}_1 = 1 \text{ and } \bm \lambda \geq \bm 0, \bm w \geq \bm 0
\end{align*}
has exactly the same objective value as the following optimization problem, provided that
one of them is bounded:
\begin{align}
    \max_{\bm h \models (\Delta, \bm N)} \min_{\bm Z \in \Sigma_{\out}} h(\bm Z)
    \label{eqn:dual:ddr:lp}
\end{align}
Moreover, Problem~\eqref{eqn:dual:ddr:lp} can be expressed as a linear program.
\end{lmm}
\begin{proof}
    Let $\Gamma_n$ denote the set of all polymatroids on $n$ variables, where $n := |\bm V|$.
As in the proof of Lemma 6.4 in~\cite{theoretics:13722}, we reformulate~\eqref{eqn:dual:ddr:lp}
as the following linear program, which proves the last statement of the lemma:
\begin{align*}
    \opt \defeq \max_{t, \bm h \in \Gamma_n} \left\{
        t \mid t \leq h(\bm Z), \forall \bm Z \in \Sigma_{\out}, \text{ and }
        h(\delta) \leq \log N_\delta, \forall \delta \in \Delta
    \right\}
\end{align*}
Let $\lambda_{\bm Z}$ and $w_\delta$ be the Lagrange multipliers associated with the
constraints $t \leq h(\bm Z)$ and $h(\delta) \leq \log N_\delta$ respectively.
The Lagrangian dual function is
\begin{align*}
\calL(\bm \lambda, \bm w) &\defeq
\max_{t, \bm h \in \Gamma_n} \left\{
    t + \sum_{\bm Z \in \Sigma_{\out}} \lambda_{\bm Z} (h(\bm Z) - t)
    + \sum_{\delta \in \Delta} w_\delta (\log N_\delta - h(\delta))
\right\} \\
&= \sum_{\delta \in \Delta} w_\delta \log N_\delta +
\max_t (1-\norm{\bm\lambda}_1) t +
\max_{\bm h \in \Gamma_n} \left\{
    \sum_{\bm Z \in \Sigma_{\out}} \lambda_{\bm Z} h(\bm Z)
    - \sum_{\delta \in \Delta} w_\delta h(\delta)
\right\}
\end{align*}
Since the optimization problem is linear, strong duality holds, and thus the Lagrangian
dual problem has the same objective as the primal problem, namely
$\opt = \min_{\bm \lambda, \bm w \geq \bm 0} \calL(\bm \lambda, \bm w)$, which we had
assumed to be bounded.

If $(\bm \lambda^*, \bm w^*)$ is an optimal solution to the dual problem, then
$\norm{\bm \lambda^*}_1 = 1$ must hold; otherwise, the objective is unbounded.
Similarly, for every $\bm h \in \Gamma_n$, we must have
\[
\sum_{\bm Z \in \Sigma_{\out}} \lambda^*_{\bm Z} h(\bm Z)
    - \sum_{\delta \in \Delta} w^*_\delta h(\delta) \leq 0
\]
in order for the objective to be bounded. This means we can reformulate the Lagrangian dual
problem to be exactly the optimization problem in the statement of the lemma.
\end{proof}

Finally, we state and prove the key supporting proposition.
\begin{prop}[~\cite{theoretics:13722,DBLP:conf/pods/Khamis0S17}]
Consider the following optimization problem
\begin{align*}
    \min_{\bm \lambda, \bm w} \quad &
    \prod_{\delta \in \Delta} N_\delta^{w_\delta} \\
    \text{s.t.} \quad &
    \sum_{\bm Z \in \Sigma_{\out}}\lambda_{\bm Z}\cdot h(\bm Z)
    \leq \sum_{\delta \in \Delta} w_\delta \cdot h(\delta)
    \text{ is a Shannon-flow inequality}\\
    & \norm{\bm \lambda}_1 = 1 \text{ and } \bm \lambda \geq \bm 0, \bm w \geq \bm 0
\end{align*}
Then, the following hold:
\begin{itemize}
    \item[(i)] The problem is solvable in polynomial time in $|2^{\bm V}|$ and $\sum_{\delta \in
    \Delta} \log N_\delta$.
    \item[(ii)] Let $(\bm \lambda, \bm w)$ be any optimal solution, and
let $B$ denote the {\em optimal} objective value. Then, for any unconditional monotonicity
term $\delta \in \Delta$, we have $N_\delta > B$ implies $w_\delta = 0$.
\end{itemize}
\label{prop:opt:ddr}
\end{prop}
\begin{proof}
Part $(i)$ follows directly from Lemma~\ref{lmm:ddr:lp}.
We prove part $(ii)$ next. Fix any optimal solution $(\bm \lambda, \bm w)$. Consider the equivalent
integral form of the inequality~\eqref{eqn:ddr:shearer} parameterized by
$(\calZ, \calD)$.
Let $m(Z)$ and $m(\delta)$ denote the multiplicities of $\bm Z \in \Sigma_{\out}$
and $\delta \in \Delta$ in $\calZ$ and $\calD$ respectively.
Then, we have $\lambda_{\bm Z} = \frac{m(\bm Z)}{|\calZ|}$ and
$w_\delta = \frac{m(\delta)}{|\calZ|}$.

Suppose to the contrary that there is a $\bar \delta \in \Delta$ such that $w_{\bar \delta} > 0$  and
$N_{\bar\delta} > B$. Then,
\[
B \geq N_{\bar\delta}^{w_{\bar \delta}} =
N_{\bar\delta}^{\frac{m(\bar \delta)}{|\calZ|}} >
B^{\frac{m(\bar \delta)}{|\calZ|}},
\]
which implies that $1 \leq m(\bar \delta) < |\calZ|$. From Lemma~\ref{lmm:reset-lemma}, we
can construct another integral Shannon-flow inequality parameterized by $(\calZ', \calD')$,
such that $\calZ' \subseteq \calZ$ and $\calD' \subseteq \calD - \{\bar \delta\}$, with
$|\calZ'| \geq |\calZ| - 1 > 0$. Let $m'(\delta)$ denote the multiplicity of $\delta$ in
$\calD'$ and $m'(\bm Z)$ denote the multiplicity of $\bm Z$ in $\calZ'$. Then, $m'(\delta)
\leq m(\delta)$ for all $\delta \neq \bar \delta$, and $m'(\bar \delta) \leq m(\bar \delta)
- 1$. Define $\bm \lambda'$, $\bm w'$ as $\lambda'_{\bm Z} \defeq \frac{m'(\bm Z)}{|\calZ'|}$
and $w'_\delta \defeq \frac{m'(\delta)}{|\calZ'|}$, then they form a feasible solution to the
optimization problem. The objective value of $(\bm \lambda', \bm w')$ is
\begin{align*}
\prod_{\delta \in \Delta} N_{\delta}^{w'_\delta}
=
\prod_{\delta \in \Delta} N_{\delta}^{\frac{m'(\delta)}{|\calZ'|}}
\leq
    \frac{
        \prod_{\delta \in \Delta } N_{\delta}^{\frac{m(\delta)}{|\calZ'|}}
    }
    {
        N_{\bar \delta}^{\frac{1}{|\calZ'|}}
    }
< \frac{B^{\frac{|\calZ|}{|\calZ'|}}}
    {B^{\frac{1}{|\calZ'|}}}
= B^{\frac{|\calZ|-1}{|\calZ'|}} \leq B,
\end{align*}
contradicting the optimality of $(\bm \lambda, \bm w)$.
\end{proof}

\subsection{Main Result: Correctness and Runtime of $\newalgorithm$}

\begin{restatable}{thm}{ThmMainResult}
    Given a disjunctive datalog rule of the form~\eqref{eq:ddr},
    input database instance over $\Sigma_{\inn}$ satisfying degree constraints
    $(\Delta, \bm N)$, and a Shannon-flow inequality
    \begin{align}
        \sum_{\bm Z \in \Sigma_{\out}} \lambda_{\bm Z} h(\bm Z)
        &\leq \sum_{\delta \in \Delta} w_\delta h(\delta).
        \label{eqn:shannon:flow:inequality:theorem}
    \end{align}
    Then, \newalgorithm can compute a model
    $\Sigma_{\out}$ for the DDR~\eqref{eq:ddr} of size $\norm{\Sigma_{\out}} = O(B)$
    in time $O(N\log N+B\log N)$ where $N$ is the size of the input database instance, and
    \begin{align}
        B \defeq \prod_{\delta \in \Delta} N_\delta^{w_\delta}.  \label{eqn:budget:B}
    \end{align}
    \label{thm:tight:panda}
\end{restatable}
\begin{proof}
Let $(\calZ, \calD)$ be the multisets representing the equivalent integral Shannon-flow
inequality~\eqref{eqn:shannon:flow:inequality:bag} of the given
inequality~\eqref{eqn:shannon:flow:inequality:theorem}.
For each $\delta  = (\bm Y | \bm X) \in \Delta$, let $R \in \Sigma_{\inn}$ be a relation where
$\max_{\bm x} \deg_R(\bm Y | \bm X = \bm x) \leq N_\delta$.
Define the following conditional sub-probability measure:
\begin{align}
    p_{\delta}(\bm y | \bm x) \defeq
    \begin{cases}
        \frac{1}{N_\delta} &
            (\bm x, \bm y) \in R_{|\bm X\cup\bm Y} \\
        0 & \text{otherwise}
    \end{cases}
    \label{eq:initial:measure}
\end{align}
These measures define the set $\calP$, which along with $(\calZ, \calD)$ are the
initial inputs to the algorithm $\newalgorithm$. The recursive calls to $\newalgorithm$ form
a tree of branching factor at most 2 at each internal node, called the {\em execution tree}.
The leaf nodes are where we return and gather the results, where for every $\bm Z \in \calZ$
the output $Q(\bm Z)$ is the union of all $Q(\bm Z)$ constructed at the leaf nodes.

{\bf Claim:} Throughout the execution, the following invariants hold:
\begin{itemize}
\item[(a)] Every pair $(\calZ, \calD)$ passed to $\newalgorithm$ defines a valid
Shannon-flow inequality of the form~\eqref{eqn:shannon:flow:inequality:bag}.
\item[(b)] For every $p_{\delta} \in \calP$ where $\delta = (\bm Y | \bm X)$, and every
$\bm x \in \dom^{\bm X}$, $p_{\bm Y|\bm X=\bm x}$ is a sub-probability measure.
\item[(c)] Every $\calZ$ passed to $\newalgorithm$ is not empty.
\item[(d)] For every $p_{\bm Y} \in \calP$ (unconditional measure), we have
$p_{\bm Y}(\bm y) > 0 \Rightarrow p_{\bm Y}(\bm y) \geq 1/B$ for all $\bm y \in \dom^{\bm Y}$.
\item[(e)] For every tuple $\bm t \in \bigjoin \Sigma_{\inn}$, there exists a leaf node
$\newalgorithm(\calZ, \calD, \calP)$ in the execution tree where
\begin{align}
    \prod_{\delta \in \calD} p_{\delta}(\bm t) \geq \frac{1}{B^{|\calZ|}}.
    \label{eqn:invariant:e}
\end{align}
\end{itemize}
Assuming the claim holds, we prove that the algorithm is correct and analyze its runtime.

We start with correctness. For any $\bm t \in \bigjoin \Sigma_{\inn}$, consider the leaf
node $\newalgorithm(\calZ, \calD, \calP)$ in the execution tree
where~\eqref{eqn:invariant:e} holds.
Since it is a leaf node, there exists $\bm Z \in \calZ \cap \calD$.
From~\eqref{eqn:invariant:e} we have $p_{\bm Z}(\bm t) >0$, which means
$\bm t_{\bm Z} \in Q(\bm Z)$. Thus, the output is a model of the input DDR.

Next, we analyze the runtime. The length of the path from the root down to the first
internal node with $2$ children is at most the length of the proof sequence of the
Shannon-flow inequality~\eqref{eqn:shannon:flow:inequality:bag}. At the branching point, the
light branch continues on with the next proof step in the sequence, while the heavy branch
resets the inequality by applying the Reset Lemma~\ref{lmm:reset-lemma}. Either way, the
potential function $|\calD| + |\calM| + 3|\calS|$ decreases by at least $1$, thanks to
Theorem~\ref{thm:proof-sequence} and Lemma~\ref{lmm:reset-lemma}. Therefore, the depth of
the execution tree is at most $|\calD| + |\calM| + 3|\calS|$, which is a query-dependent
parameter.

To show the $O((N+B)\log N)$-runtime overall, it is thus sufficient to argue that
every call to $\applystep$ takes $O((N+B)\log N)$-time.
From invariant (d), every unconditional measure $p_{\bm Y} \in \calP$ has
at most $B$ tuples in its support. Therefore, every decomposition and monotonicity step
can be implemented in $O((N+B)\log N)$-time because computing the marginal and conditional
measures can be done with at most $2$ passes over the support of the original measure.
The submodularity step is even simpler because it only requires renaming a measure.
Finally, the composition step~\eqref{eqn:truncated:product} can also be implemented in
$O((N+B)\log N)$-time because we can scan over $p_{\bm X}$, and for each $\bm x$ in its
support, compute the threshold $1/(B \cdot p_{\bm X}(\bm x))$ to filter the support of
$p_{\bm Y | \bm X=\bm x}$.
The output measure $p_{\bm X \bm Y}$ has at most $B$ tuples in its support because of
invariant (d).

It remains to prove the claim.  Invariants (a) and (b) hold trivially by design. Invariant
(c) holds at the root of the execution tree because $\norm{\bm\lambda}_1=1$, and it
continues to hold afterwards, thanks to the condition $|\calZ| > 1$ in
line~\eqref{alg:tight:panda:reset} of Algorithm~\ref{alg:tight:panda}.

Invariant (d) also holds by design for every new unconditional measure $p_{\bm Y}$ created
during the execution, from the composition step. Thus, we only need to show that it holds at
the start of the algorithm. Without loss of generality, we can assume that the coefficients
$(\bm w, \bm \delta)$ given in the input Shannon-flow
inequality~\eqref{eqn:shannon:flow:inequality:theorem} are an optimal solution to the
optimization problem in Proposition~\ref{prop:opt:ddr}, which implies we can assume
$\degree_{\Sigma_{\inn}}(\delta) \leq B$ for every unconditional term $\delta \in \Delta$.
Thus, invariant (d) holds at the start.

We prove invariant (e) by induction on the execution tree.  We will prove that (e) holds for
{\em every} current execution tree during the execution of the algorithm, not just the last
one. At the start, there is only one initial node, where inequality~\eqref{eqn:invariant:e}
holds as an {\em equality} for every tuple $\bm t \in \bigjoin \Sigma_{\inn}$. It is
straightforward to verify that the product $\prod_{\delta \in \calD} p_{\delta}(\bm t)$ is
not reduced by any step other than the composition step in Eq.~\eqref{eqn:truncated:product}. In the composition step $h(\bm
X)+h(\bm Y|\bm X) \to h(\bm X\bm Y)$, the node $(\calZ, \calD, \calP)$ branches into two
child nodes: the light child node $(\calZ, \calD^\ell, \calP^\ell)$ and the heavy child node
$(\calZ^h, \calD^h, \calP^h)$, where the heavy child node only exists assuming $|\calZ|>1$. Consider a tuple $\bm t$ for which~\eqref{eqn:invariant:e}
holds at the parent node. We consider two cases:
\begin{itemize}
    \item If $p_{\bm X}(\bm t_{\bm X}) \cdot p_{\bm Y | \bm X}(\bm t_{\bm Y} | \bm t_{\bm X}) \geq
1/B$, then $\prod_{\delta \in \calD^\ell} p_{\delta}(\bm t) = \prod_{\delta \in \calD} p_{\delta}(\bm t)$,
and thus~\eqref{eqn:invariant:e} holds at the light child node.
    \item If $p_{\bm X}(\bm t_{\bm X}) \cdot p_{\bm Y | \bm X}(\bm t_{\bm Y} | \bm t_{\bm X}) <
1/B$, then we want to show that $\prod_{\delta \in \calD^h} p_{\delta}(\bm t) \geq
1/B^{|\calZ^h|}$ for the $(\calZ^h, \calD^h)$ pair resulted from the reset step,
where we know $\calZ^h \subseteq \calZ$, $|\calZ^h| \geq |\calZ|-1$,
and $\calD^h \subseteq \calD - \{(\bm X), (\bm Y | \bm X)\}$.
We also want to show that $|\calZ| > 1$ since, otherwise, the heavy child would not have been created in the first place.
Since~\eqref{eqn:invariant:e} holds at the parent node, we have
    \begin{align}
        1 \geq \prod_{\delta \in \calD^h} p_{\delta}(\bm t)
        \geq \frac{\prod_{\delta \in \calD} p_{\delta}(\bm t)}
        {p_{\bm X}(\bm t_{\bm X}) \cdot p_{\bm Y | \bm X}(\bm t_{\bm Y} | \bm t_{\bm X})}
        > B \cdot \frac{1}{B^{|\calZ|}}
        \geq \frac{1}{B^{|\calZ^h|}}.
        \label{eq:invariant:heavy}
    \end{align}
This inequality proves invariant (e) for the heavy child node, assuming $|\calZ| > 1$. It
also proves that $|\calZ| > 1$ because if $|\calZ| = 1$, then the inequality becomes $1 >
1$, which is a contradiction. (Invariant (c) rules out the case of $|\calZ| = 0$.)
\end{itemize}

\end{proof}

\section{Disjunctive datalog rules, conjunctive queries, and the submodular width}
\label{sec:subw:ddr}

This section provides a brief overview of the notions of fractional hypertree width and
submodular width under degree constraints, and how conjunctive queries can be answered in
submodular width time via the connection to DDRs. All the materials were implicit or
explicit in~\cite{DBLP:conf/pods/Khamis0S17,theoretics:13722}. However, we adapted the
definitions to the bounds presented in this paper which are expressed directly in the
input degree constraints, so they are in the primal form instead of the dual form as
in~\cite{theoretics:13722,DBLP:conf/pods/Khamis0S17}.

\subsection{Tree decompositions and free-connex tree decompositions}

A {\em tree decomposition} of a hypergraph $\calH=(\calV,\calE)$ is a pair $(T, \chi)$
where $T=(V(T),E(T))$ is a tree and $\chi: V(T) \to 2^{\calV}$ is a mapping that
assigns to each node of $T$ a set of vertices of $\calH$, called a \emph{bag}, such that:
(i) for every hyperedge $S \in \calE$, there exists a node $t \in V(T)$ such that
$S \subseteq \chi(t)$; and (ii) for every vertex $v \in \calV$, the set of nodes
$\{t \in V(T) \mid v \in \chi(t)\}$ induces a connected subtree of $T$.

Consider a conjunctive query (CQ)
$Q(F) \text{ :- } \bigwedge_{S \in \calE} R_S(S)$
where $\calH = (\calV, \calE)$ is its associated hypergraph, $F \subseteq \calV$
is the set of free variables of $Q$, and each $R_S$ is a relation over the set of
variables $S \subseteq \calV$. Note that we abuse notation and use $S$ to refer to both
the hyperedge in $\calE$ and the set of variables in the relation $R_S$.
In the notations established in the paper, $\Sigma_\inn = \{R_S \mid S \in \calE\}$.

A {\em free-connex} tree decomposition~\cite{DBLP:conf/csl/BaganDG07}
of $Q$ is a tree decomposition $(T, \chi)$ of its hypergraph $\calH$ such that there is a connected subtree of $T$ whose bags contain {\em all} the free variables $F$ of $Q$ and {\em no} non-free variables.
Note that, if $Q$ was a Boolean CQ or a full CQ,
then every tree decomposition of $\calH$ is free-connex. Let $\td(Q)$ denote the set of
all free-connex tree decompositions of $Q$.

\subsection{Fractional hypertree width under degree constraints}

Suppose $\Sigma_\inn$ satisfies degree constraints $(\Delta, \bm N)$.
A free-connex tree decomposition $(T,\chi)$ of $Q$ can be thought of as a query plan for answering
$Q$. Each bag $\chi(t)$ defines the following DDR:
\begin{align}
\label{eq:ddr-bag}
Q_t(\chi(t)) \text{ :- } \bigwedge_{S \in \calE} R_S(S).
\end{align}
The answer to $Q$ can be computed in linear time (and with constant delay)
from the answers to all DDRs~\eqref{eq:ddr-bag} (which are CQs)~\cite{DBLP:conf/csl/BaganDG07},
one for each bag of $(T, \chi)$.
The time it takes to answer the DDR~\eqref{eq:ddr-bag}, from Theorem~\ref{thm:tight:panda} is
$O((N+B)\log N)$ where
\begin{align}
B \defeq \min_{\bm w\geq \bm 0} \left\{
    \prod_{\delta \in \Delta} N_\delta^{w_\delta}
    \text{ where }
    h(\chi(t)) \leq \sum_{\delta\in \Delta} w_{\delta} h(\delta)
    \text{ is a valid Shannon inequality}
\right\}
\end{align}
To compute this bound, we can solve equivalently the following optimization problem
(which is solvable in polynomial time, thanks to Proposition~\ref{prop:opt:ddr}):
\begin{align}
\label{eq:rho:star}
\rho^*(t, \Delta, \bm N) &\defeq
    \min_{\bm w\geq \bm 0} \sum_{\delta \in \Delta} w_\delta \cdot \log_N N_\delta \\
\text{subject to } & h(\chi(t)) \leq \sum_{\delta\in \Delta} w_{\delta} h(\delta)
    \text{ is a valid Shannon inequality} \label{eqn:new:shearer}
\end{align}
The quantity $\rho^*(t, \Delta, \bm N)$ is the generalization of the {\em fractional edge cover number}
of a bag $t$ under degree constraints.
If all constraints were cardinality constraints,
then~\eqref{eqn:new:shearer} is exactly the Shearer's inequality.

Hence, if our query plan consists of a single tree decomposition, then the best time
we can hope for is parameterized by the generalized fractional hypertree width of $Q$ under
degree constraints:
\begin{align}
\fhtw(Q, \Delta, \bm N) \defeq \min_{(T,\chi) \in \td(Q)} \max_{t \in V(T)} \rho^*(t, \Delta, \bm N).
\label{eqn:fhtw:dc}
\end{align}
Note that $\td(Q)$ is the set of all free-connex tree decompositions of $Q$, and thus
this notion generalizes the classical fractional hypertree width to not only handle
degree constraints but also non-Boolean and non-full CQs.
A direct corollary of Theorem~\ref{thm:tight:panda} is the following:

\begin{cor}
\label{cor:fhtw}
Given a CQ $Q$ and degree constraints $(\Delta, \bm N)$ on its input relations, we can
compute the answer to $Q$ in time $O(N^{\fhtw(Q, \Delta, \bm N)} \log N + |Q|)$, with
constant-delay enumeration of the answers.
\end{cor}

There is a subtle point that our definition of $\fhtw$ in~\eqref{eqn:fhtw:dc} above is a
smooth generalization of the original fractional hypertree width definition
from~\cite{DBLP:conf/soda/GroheM06} where the fractional edge covering number is replaced by
$\rho^*(t, \Delta, \bm N)$. However, the definition of $\fhtw$ given
in~\cite{DBLP:conf/pods/Khamis0S17,theoretics:13722} is slightly different because it is
based on the dual form of the optimization problem.\footnote{As a special case, the fractional
edge cover number is the same as the fractional vertex packing number.}
Nevertheless, thanks to Lemma~\ref{lmm:ddr:lp}, we can see that the definitions are equivalent:

\begin{align*}
\fhtw(Q, \Delta, \bm N) \defeq
    \min_{(T,\chi) \in \td(Q)} \max_{t \in V(T)} \rho^*(t, \Delta, \bm N)
    =
    \min_{(T,\chi) \in \td(Q)} \max_{t \in V(T)}
    \max_{\bm h \models (\Delta, \bm N)} h(\chi(t))
\end{align*}

\subsection{Submodular width under degree constraints}

The fractional hypertree width was derived by minimizing the output size, over all tree decompositions
$(T, \chi)$ of the following query:
\begin{align}
    \bigwedge_{t \in V(T)} Q_t(\chi(t)) \text{ :- } \bigwedge_{S \in \calE} R_S(S).
\end{align}
The conjunction in the head means we can evaluate each DDR $Q_t$ independently. Daniel
Marx~\cite{DBLP:journals/jacm/Marx13} had a brilliant insight that we do not need to use a
single tree decomposition, even the best one; instead, we can load-balance the computation
over multiple tree decompositions to obtain a better runtime. To do so, we want to answer
the following query optimally:
\begin{align}
    \bigvee_{(T, \chi) \in \td(Q)} \bigwedge_{t \in V(T)} Q_t(\chi(t)) \text{ :- } \bigwedge_{S \in \calE} R_S(S).
    \label{eqn:multiple:td}
\end{align}
Semantically, query~\eqref{eqn:multiple:td} means that, for every satisfying tuple
of the query's body, as long as there is one tree decomposition for which all their
bag atoms are satisfied, then the tuple is in the output.
If we can compute an answer to query~\eqref{eqn:multiple:td}, then from each tree decomposition,
we can enumerate with constant delay their contribution to the answer to $Q$.
This means over all, the answer to $Q$ can also be enumerated with constant delay.

Query~\eqref{eqn:multiple:td} is neither a typical CQ nor a typical DDR. The second idea
we need is the distributivity of $\vee$ over $\wedge$ to rewrite it from DNF-form to CNF-form.
In order to describe this rewriting, we need a numbering of tree decompositions and their bags.
We number the tree decompositions in $\td(Q)$ as $(T_i, \chi_i)$ for $i \in [m]$,
and let $\bm T$ denote the set of all tuples $\bm t = (t_1, \ldots, t_m)$ where
$t_i \in V(T_i)$ is a bag in tree decomposition $T_i$.
Then, we can rewrite query~\eqref{eqn:multiple:td} as:
\begin{align}
    \bigwedge_{\bm t \in \bm T} \bigvee_{i \in [m]} Q_{t_i}(\chi_i(t_i)) \text{ :- } \bigwedge_{S \in \calE} R_S(S).
\end{align}
This query can be answered by answering $|\bm T|$ many DDRs of the form:
\begin{align}
    \bigvee_{i \in [m]} Q_{t_i}(\chi_i(t_i)) \text{ :- } \bigwedge_{S \in \calE} R_S(S).
\end{align}
From Theorem~\ref{thm:tight:panda}, the time to answer this DDR is
$O((N+B)\log N)$ where
\begin{align}
B \defeq \min_{\bm\lambda, \bm w\geq \bm 0} \left\{
    \prod_{\delta \in \Delta} N_\delta^{w_\delta}
    \text{ s.t. }
    \sum_{i \in [m]} \lambda_i \cdot h(\chi_i(t_i)) \leq \sum_{\delta\in \Delta} w_{\delta} h(\delta)
    \text{ is a valid Shannon inequality}
\right\}
\end{align}
To compute this bound, we can solve equivalently the following optimization problem,
which is solvable in polynomial time, thanks to Proposition~\ref{prop:opt:ddr}:
(Note that, compared to Eq.~\eqref{eq:rho:star}, we now have an additional variable vector
$\bm \lambda$.)
\begin{align}
\label{eq:sw:star}
\rho^*(\bm t, \Delta, \bm N) &\defeq
    \min_{\bm\lambda, \bm w\geq \bm 0} \sum_{\delta \in \Delta} w_\delta \cdot \log_N N_\delta \\
\text{subject to } & \sum_{i \in [m]} \lambda_i \cdot h(\chi_i(t_i)) \leq \sum_{\delta\in \Delta} w_{\delta} h(\delta)
    \text{ is a valid Shannon inequality} \label{eqn:new:shearer:sw}
\end{align}
Since we have to answer all DDRs indexed by $\bm t \in \bm T$, the time spent is dominated
by the worst combination. The submodular width of $Q$ under degree constraints is defined
as:
\begin{align}
\subw(Q, \Delta, \bm N) \defeq \max_{\bm t \in \bm T} \rho^*(\bm t, \Delta, \bm N).
\end{align}
A direct corollary of Theorem~\ref{thm:tight:panda} is the following:
\begin{cor}
\label{cor:subw}
Given a CQ $Q$ and degree constraints $(\Delta, \bm N)$ on its input relations, we can
compute the answer to $Q$ in time $O(N^{\subw(Q, \Delta, \bm N)} \log N + |Q|)$, with
constant-delay enumeration of the answers.
\end{cor}

Similar to the fractional hypertree width case, our definition of the submodular width under
degree constraints here is slightly different from that
in~\cite{DBLP:conf/pods/Khamis0S17,theoretics:13722,DBLP:journals/jacm/Marx13}, because
we use the primal form of the optimization problem.
Thanks to Lemma~\ref{lmm:ddr:lp}, we can see that the definitions are equivalent:
\begin{align*}
\subw(Q, \Delta, \bm N) &\defeq
    \max_{\bm t \in \bm T} \rho^*(\bm t, \Delta, \bm N)=
    \max_{\bm t \in \bm T}
    \max_{\bm h \models (\Delta, \bm N)}
    \min_{i \in [m]} h(\chi_i(t_i)) \\
    &=
    \max_{\bm h \models (\Delta, \bm N)}
    \max_{\bm t \in \bm T}
    \min_{i \in [m]} h(\chi_i(t_i))=
    \max_{\bm h \models (\Delta, \bm N)}
    \min_{(T,\chi) \in \td(Q)}
    \max_{t \in V(T)} h(\chi(t))
\end{align*}

\section{Extensions to handle $\ell_p$-norm constraints}
\label{sec:lp:norm}

Abo Khamis et al~\cite{DBLP:journals/pacmmod/KhamisNOS24} showed that the PANDA framework
can be extended to handle a more general class of constraints, called {\em $\ell_p$-norm
constraints}. Let $p \in (0, \infty]$ be a positive integer. Let $\theta = (\bm Y | \bm X)_p$
denote what we call a {\em $\ell_p$-norm term} (to generalize monotonicity term).
Let $R$ be an input relation over schema $\Sigma$.
Referring back to notations defined in~\eqref{eq:degree:x} and~\eqref{eq:degree:delta},
we define
\begin{align}
\deg_R(\theta) &:= \norm{\deg_R(\bm Y | \bm X = \bm x)}_p
&& \deg_\Sigma(\theta) := \min_{R \in \Sigma} \deg_R(\theta) \label{eq:degree:nu}
\end{align}
Given a database instance over schema $\Sigma$, we say that $\Sigma$ satisfies the {\em $\ell_p$-norm
constraint} $(\theta, N_\theta)$, and write $\Sigma \models (\theta, N_\theta)$, if $\deg_\Sigma(\theta) \leq N_\theta$.
Similar to $h(\delta)$ and $h(\sigma)$, we define
\begin{align}
    h(\theta) &:= \frac 1 p h(\bm X) + h(\bm Y | \bm X)
    && \theta = (\bm Y | \bm X)_p. \label{eq:h:nu}
\end{align}
Note that if $p = \infty$ then $h(\theta) = h(\delta)$ and $\deg_\Sigma(\theta) =
\deg_\Sigma(\delta)$ where $\delta = (\bm Y | \bm X)$; in particular, $\ell_p$-norm
constraints strictly generalize degree constraints. Let $\Theta$ denote a set of
$\ell_p$-norm terms (where $p$ can also vary). Let $\bm N = (N_\theta)_{\theta \in \Theta}$
be a vector of positive real numbers. Given a set $\Sigma_{\inn}$ of input relations, we say
that $\Sigma_{\inn} \models (\Theta, \bm N)$ if
$\Sigma_\inn \models (\theta, N_\theta)$
for every $\theta \in \Theta$.
The following analog of Theorem~\ref{th:upper:bound} holds for $\ell_p$-norm constraints:

\begin{cor}[\cite{DBLP:journals/pacmmod/KhamisNOS24}]
\label{cor:upper:bound:lp:norm}
Consider a DDR of the form~\eqref{eq:ddr} with input schema $\Sigma_{\inn}$, and output
schema $\Sigma_{\out}$ such that $\Sigma_{\inn} \models (\Theta, \bm N)$. Assume that there
exist two {\em non-negative} rational weight vectors $\bm w \defeq (w_\theta)_{\theta \in
\Theta}$ and $\bm \lambda \defeq (\lambda_{\bm Z})_{\bm Z \in\Sigma_{\out}}$ with $\norm{\bm
\lambda}_1=1$, where the following is a Shannon-flow inequality:
\begin{align}
    \sum_{\bm Z \in \Sigma_{\out}}\lambda_{\bm Z}\cdot h(\bm Z)
    & \leq \sum_{\theta \in \Theta} w_\theta \cdot h(\theta).
    \label{eqn:shannon:flow:inequality:lp:norm}
\end{align}
Then, for any input instance $\Sigma_{\inn}$ of the DDR~\eqref{eq:ddr},
there exists a model $\Sigma_{\out}$ for the DDR for which:
  {
  \begin{align}
    \max_{\bm Z \in \Sigma_{\out}} |Q(\bm Z)| &\leq
      \prod_{\theta \in \Theta} N_\theta^{w_\theta}
  \end{align}
  }
\end{cor}
\begin{proof}
The proof is virtually identical to that of Corollary~\ref{cor:upper:bound}, where we apply
Lemma~\ref{lmm:probabilistic:inequality} but with a different initialization
of the probability measures.
In particular, inequality~\eqref{eqn:shannon:flow:inequality:lp:norm} {\em is} a
Shannon-flow inequality of the form~\eqref{eqn:ddr:shearer} when we expand it out
using the definition of $h(\theta)$ in Eq.~\eqref{eq:h:nu}:
\begin{align}
\sum_{\bm Z \in \Sigma_{\out}}\lambda_{\bm Z}\cdot h(\bm Z)
\leq \sum_{\theta \in \Theta} w_\theta \cdot h(\theta)
= \sum_{(\bm Y | \bm X)_p \in \Theta} \frac{w_\theta}{p} \cdot \left( h(\bm X) + p \cdot h(\bm Y | \bm X) \right).
\label{eq:lp-norm:shannon:flow:inequality}
\end{align}
For each $\theta = (\bm Y | \bm X)_p \in \Theta$, let $R \in \Sigma_{\inn}$ be a relation
such that $\deg_R(\theta) = \deg_\Sigma(\theta)$. We define two
sub-probability measures $p_{\bm X}$ and $p_{\bm Y | \bm X}$ as follows:
\begin{align}
p_{\bm X}(\bm x) &:= \frac{(\deg_R(\bm Y| \bm X=\bm x))^p}{\norm{\deg_R(\bm Y| \bm X)}^p_p}
&& p_{\bm Y | \bm X}(\bm y | \bm x)
:= \frac{1}{\deg_R(\bm Y | \bm X=\bm x)}
\label{eq:lp-norm:initialization}
\end{align}
Define $B := \prod_{\theta \in \Theta} N_\theta^{w_\theta}$, then from
Lemma~\ref{lmm:probabilistic:inequality} there exist coefficients
$\bm \lambda = (\lambda_{\bm Z})_{\bm Z \in \Sigma_{\out}}$ with $\norm{\bm\lambda}_1=1$ and sub-probability measures $p_{\bm Z}$,
one for each $\bm Z \in \Sigma_{\out}$, such that
\begin{align*}
\prod_{\bm Z \in \Sigma_{\out}} p_{\bm Z}(\bm t)^{\lambda_{\bm Z}}
\geq \prod_{\theta = (\bm Y|\bm X)_p \in \Theta}
 [p_{\bm X}(\bm t) \cdot p_{\bm Y | \bm X}(\bm t)^p]^{\frac{w_\theta}{p}}
 &\geq \prod_{\theta \in \Theta} \left(\frac{1}{N_\theta^p} \right)^{w_\theta/p} = \frac 1 B.
\end{align*}
The rest of the proof is identical to that of Corollary~\ref{cor:upper:bound}.
\end{proof}

The proof above also contains the initialization of the sub-probability measures
that allows $\newalgorithm$ to handle $\ell_p$-norm constraints, proving
the following analog of Theorem~\ref{thm:tight:panda}:

\begin{cor}
    Given a disjunctive datalog rule of the form~\eqref{eq:ddr},
    input database instance over $\Sigma_{\inn}$ satisfying $\ell_p$-norm constraints
    $(\Theta, \bm N)$, and a Shannon-flow inequality
    $\sum_{\bm Z \in \Sigma_{\out}} \lambda_{\bm Z} h(\bm Z)
    \leq \sum_{\theta \in \Theta} w_\theta h(\theta)$.
    Then, \newalgorithm can compute a model
    $\Sigma_{\out}$ for the DDR~\eqref{eq:ddr} of size $\norm{\Sigma_{\out}} = O(B)$
    in time $O(N\log N+B\log N)$ where $N$ is the size of the input database instance, and
    $B \defeq \prod_{\theta \in \Theta} N_\theta^{w_\theta}$.
    \label{cor:tight:panda:lp:norm}
\end{cor}

\section{Conclusion}
\label{sec:conclusion}

The $\newalgorithm$ algorithm has two key aspects. The first is the idea of leveraging
information theoretic inequalities to derive high-level query plans; and the second is to
collect precise skewness information during the execution of the query plan in order to
partition the data based on general hyperplane cuts to avoid skews. The partitioning
strategies are designed to load-balance the computation among all sub-query plans. Since we
were interested in meeting the worst-case bounds defined by the submodular width, we used
the threshold parameter $B$ to guide the partitioning of the data.

We believe that this idea of load-balancing data processing amongst different sub-query
plans is a powerful and practical idea. An interesting future direction to explore is to see
if the load-balancing can be more adaptive to the actual data distribution.

The key remaining open question at the center of the $\panda$ framework is to bound the length
of the proof sequence needed to prove a given Shannon-flow inequality. This problem is
closely related to the computational complexity of the polymatroid bound question, which is
also open. There are some sub-classes of degree constraints under which a polynomial-length
proof sequence exists, and those are {\em also} the cases where the polymatroid bound is
computable in PTIME~\cite{DBLP:journals/pacmmod/ImMNP25}.

Both $\fhtw$ and $\subw$ are defined as optimization problems over the set $\td(Q)$ which
is of exponential size in query complexity. It is easy to see that, for some queries we do
not need to optimize over {\em all} tree decompositions, but only a small subset of them.
It would be interesting to characterize the class of queries for which optimizing over
a small subset of tree decompositions suffices to compute $\fhtw$ and $\subw$.

Another important question is to study, for which class of input, this framework can be used
to answer aggregate queries efficiently. There have been some attempts at defining the
counting version of the submodular width~\cite{10.1145/3426865}; however, that width is
higher than necessary for some class of queries, as shown
in~\cite{DBLP:conf/stoc/BringmannG25}.

\section*{Acknowledgments}
This work was partially supported by NSF IIS 2314527, NSF SHF 2312195, NSF III 2507117, and a Microsoft professorship.

\bibliographystyle{siam}
\bibliography{bib}

\appendix
\section{Proofs of Some Supporting Results on Proof Sequences}
\label{app:supporting:results}

To make this paper more self-contained, we include here the proofs of two supporting results
on proof sequences, which were not explicitly stated in~\cite{theoretics:13722,DBLP:conf/pods/Khamis0S17}.

\ThmProofSequence*

\begin{proof}[Proof of Theorem~\ref{thm:proof-sequence}]
Let $(\calZ, \calD)$ be parameters of the integral Shannon-flow inequality
\eqref{eqn:shannon:flow:inequality} witnessed by $(\calM, \calS)$.
We start with $i=0$: $(\calD_0, \calM_0, \calS_0) := (\calD, \calM, \calS)$,
and we will iteratively construct $(\calD_{i+1}, \calM_{i+1}, \calS_{i+1})$.

By Proposition~\ref{prop:shannon:identity}, identity~\eqref{eqn:bag:identity} holds.
If $\calZ \subseteq \calD_i$ (in the multiset sense), then we are done; no proof step is needed.
So let's assume $\calZ \not\subseteq \calD_i$.

We first claim that there must be at least one unconditional term $\bm W \in \calD_i$
that is {\em not} in $\calZ$. Assume for the
sake of contradiction that there is no unconditional term in $\calD_i$. Set all (symbolic)
variables $h(\bm X) := 1$, $\bm X \neq \emptyset$. Then,
$h(\bm Y| \bm X)$ is $0$ if $\bm X \neq \emptyset$, and $1$ otherwise.
Moreover, $h(\bm Y| \bm X)\geq 0$ and $h(\bm Y; \bm Z | \bm X) \geq 0$.
Hence, the LHS of~\eqref{eqn:bag:identity} is exactly $|\calZ|$ and the RHS is at most the number
of unconditional terms in $\calD_i$.
Therefore, the number of unconditional terms in $\calD_i$ must be at least $|\calZ|$.
If every unconditional term in $\calD_i$ is also in $\calZ$, then $\calZ$ must be the set of unconditional terms in $\calD_i$. This contradicts our assumption that $\calZ \not\subseteq \calD_i$.

Next, take any $\bm W \in \calD_i$ that is not in $\calZ$. Since
identity~\eqref{eqn:bag:identity} holds, there must be some term that cancels $h(\bm W)$.
We consider a few cases:
\begin{itemize}
    \item If $(\bm Y | \bm W) \in \calD_i$ so that $h(\bm Y | \bm W)$ cancels $h(\bm W)$,
    then we define
    \begin{align*}
        \calD_{i+1} := \calD_i \setminus \{ \bm W, (\bm Y | \bm W) \} \cup \{ \bm Y \bm W\},\quad\quad
        \calM_{i+1} := \calM_i \quad\quad
        \calS_{i+1} := \calS_i
    \end{align*}
    and add to the proof sequence the composition
    step $h(\bm W) + h(\bm Y | \bm W) \to h(\bm Y \bm W)$.
    \item If $(\bm Y | \bm X) \in \calM_i$  where $\bm W = \bm Y\bm X$ so that $-h(\bm Y | \bm X)$ cancels $h(\bm W)$,
    then we define
    \begin{align*}
        \calD_{i+1} := \calD_i \setminus \{ \bm W \} \cup \{ \bm X \},\quad\quad
        \calM_{i+1} := \calM_i \setminus \{ (\bm Y | \bm X) \},\quad\quad
        \calS_{i+1} := \calS_i
    \end{align*}
    and add to the proof sequence the monotonicity
    step $h(\bm W) \to h(\bm X)$.
    \item If $(\bm Y; \bm Z | \bm X) \in \calS_i$ so that
    $-h(\bm Y; \bm Z | \bm X)$ cancels $h(\bm W)$, where $\bm W = \bm X\bm Y$,
    then we define
    \begin{align*}
        \calD_{i+1} := \calD_i \setminus \{ \bm W \} \cup \{ (\bm X), (\bm Y | \bm X\bm Z) \} ,\quad\quad
        \calM_{i+1} := \calM_i ,\quad\quad
        \calS_{i+1} := \calS_i \setminus \{ (\bm Y; \bm Z | \bm X) \}
    \end{align*}
    and add to the proof sequence two steps:
    the first is a decomposition step
    $h(\bm W) \to h(\bm X) + h(\bm Y | \bm X)$,
    and the second is a submodularity step
    $h(\bm Y | \bm X) \to h(\bm Y | \bm X\bm Z)$.
\end{itemize}
In all cases above, the potential function $|\calD|+|\calM|+3|\calS|$ decreases by at least
$1$ per every step added.
\end{proof}

\LmmReset*

\begin{proof}[Proof of Lemma~\ref{lmm:reset-lemma}]
This proof is very similar to the proof of Theorem~\ref{thm:proof-sequence} above.
Let $\bm W \in \calD$ be any unconditional term. Since~\eqref{eqn:bag:identity} holds
as an identity, there must be some term that cancels $h(\bm W)$. We consider a few cases:
\begin{itemize}
    \item $\bm W \in \calZ$, in this case we remove $\bm W$ from $\calD$ and from $\calZ$
    \item $(\bm Y | \bm W) \in \calD$ so that $h(\bm Y | \bm W)$ cancels $h(\bm W)$,
    then we apply a composition step $h(\bm W) + h(\bm Y | \bm W) \to h(\bm Y\bm W)$, which means
    we set $$\calD' := \calD \setminus \{ \bm W, (\bm Y | \bm W) \} \cup \{ \bm Y \bm W \}$$
    and then apply induction to remove $\bm Y\bm W$ from $\calD'$.
    \item $(\bm Y | \bm X) \in \calM$ where $\bm W = \bm X \bm Y$ so that $-h(\bm Y | \bm X)$ cancels $h(\bm W)$,
    then we apply a monotonicity step $h(\bm W) \to h(\bm X)$, which means
    we set
    \begin{align*}
        \calD' := \calD \setminus \{ \bm W \} \cup \{ \bm X\} ,\quad\quad
        \calM' := \calM \setminus \{ (\bm Y | \bm X) \}
    \end{align*}
    while keeping $\calS$ unchanged. Then, we apply induction to remove $\bm X$ from $\calD'$.
    \item $(\bm Y; \bm Z | \bm X) \in \calS$ so that
    $-h(\bm Y; \bm Z | \bm X)$ cancels $h(\bm W)$, where $\bm W = \bm X\bm Y$;
    then note that
    \[
        h(\bm W) - h(\bm Y; \bm Z | \bm X) = h(\bm X \bm Y) -
        \big( h(\bm X \bm Y) + h(\bm X \bm Z) - h(\bm X) - h(\bm X\bm Y \bm Z)\big)
        = h(\bm X \bm Y \bm Z) - h(\bm Z | \bm X).
    \]
    In this case we set
    \begin{align*}
        \calD' := \calD \setminus \{ \bm W \} \cup \{ (\bm X\bm Y \bm Z) \} ,\quad\quad
        \calM' := \calM \cup \{ (\bm Z | \bm X) \},\quad\quad
        \calS' := \calS \setminus \{ (\bm Y; \bm Z | \bm X) \}
    \end{align*}
    and then apply induction to remove $\bm X\bm Y \bm Z$ from $\calD'$.
\end{itemize}
Note that in each of the cases above, the potential function $|\calD| + |\calM| + 2|\calS|$
decreases by at least $1$ per every step.
\end{proof}

\section{Further examples and discussion}
\label{app:examples}

In Section~\ref{sec:motivations} we justified the need to partition the space of
polymatroids using hyperplanes that are not axis parallel; our algorithm $\newalgorithm$
achieved this by keeping track of probabilities, allowing it to determine this partition
dynamically, during query execution.

Our discussion in this section has two goals.  The first is to illustrate more convincing
DDR examples than the hexagon query from Section~\ref{sec:motivations} that require
hyperplane partitions.  The second is to discuss cases when $\newalgorithm$ can be
simplified, by describing directly the multi-join query that computes each target of the
DDR: this is possible when the associated Shannon inequality is a ``sum of chains", which we
define next.

A \emph{chain inequality} is a Shannon inequality of the form:
\begin{align*}
  h(\bm Z) \leq & \sum_{i=1,k} h(\bm Y_i|\bm X_i)
\end{align*}
where $\bm X_i \subseteq \bigcup_{j < i} \bm Y_j$, and $\bm X_j \cap \bm Y_i = \emptyset$
for all $j \leq i$, A \emph{sum of chains} is a Shannon-flow
inequality~\eqref{eqn:shannon:flow:inequality} that is the sum of chain inequalities: we
require the LHS of~\eqref{eqn:shannon:flow:inequality} to be precisely the set of LHS terms
of all the chain inequalities, while the RHS of~\eqref{eqn:shannon:flow:inequality} is
required to be obtained from the RHS expressions of the chain inequalities by using only
composition steps: $h(\bm V|\bm U)+h(\bm U) \rightarrow h(\bm U\bm V)$.  For simple
illustrations, both inequalities discussed in Section~\ref{sec:motivations} are sums of
chains: inequality~\eqref{eq:q1:inequality} is the sum of the two chains
in~\eqref{eq:q1:chains}, while inequality~\eqref{eq:q:hex:inequality} is the sum of the
chains in~\eqref{eq:q:hex:chains}.

In the examples below we will refer to a polymatroid $h^*$ as a \emph{tight polymatroid} for
the Shannon-flow inequality~\eqref{eqn:shannon:flow:inequality} if the inequality becomes an
equality when replacing $h$ with $h^*$.

\begin{example} \label{ex:discussion:1}
  Consider the following DDR (we drop the comma from all atoms to
  reduce clutter):
  \begin{align*}
    U(A_0A_1A_2B_1) &\vee V(B_0B_1B_2C_1)\vee W(C_0C_1C_2A_1) \\
  &=  R_1(A_0A_1)\wedge R_2(A_1A_2) \wedge S_1(B_0B_1)\wedge  S_2(B_1B_2) \wedge T_1(C_0C_1)\wedge T_2(C_1C_2)
  \end{align*}
  Assume that all input relations have size $\leq N$.  We claim that we can compute an
  output of the DDR satisfying $\max(|U|,|V|,|W|) \leq N^2$.  Unlike the simple hexagon
  query in Section~\ref{sec:motivations}, the Generic Join algorithm no longer applies here,
  because no single variable occurs in all targets of the DDR.  To prove the claim, we will
  start by proving the following Shannon inequality:
  \begin{align*}
    h(A_0A_1A_2B_1)&+h(B_0B_1B_2C_1)+h(C_0C_1C_2A_1) \\
    & \leq h(A_0A_1)+h(A_1A_2) + h(B_0B_1)+h(B_1B_2) + h(C_0C_1)+h(C_1C_2)
  \end{align*}
  The modular function $h^*$ defined by $h^*(\bm X) = |\bm X|/2$ for
  every set $\bm X \subseteq \set{A_0, A_1, \ldots, C_2}$ is edge
  dominated, and also tight for the inequality above.  The inequality
  holds because it is the sum of the following chain inequalities:
  \begin{align*}
    h(A_0A_1A_2B_1) &\leq h(A_0A_1)+h(A_2|A_1)+h(B_1) \\
    h(B_1B_2B_3C_1) &\leq h(B_0B_1)+h(B_2|B_1)+h(C_1) \\
    h(C_0C_1C_2A_1) &\leq h(A_1) + h(C_0C_1) + h(C_2|C_1)
  \end{align*}
  Evidently, $h^*$ is tight for each of the three chain inequalities.
  Before we derive an algorithm for computing the DDR, we describe the
  space partition of the edge-dominated polymatroids $h$ into
  $\calP_1\cup \calP_2 \cup \calP_3$ such that, in each partition, one
  of the three inequalities above has the LHS $\leq 2$.  In fact, the
  three sets can be described easily by:
  \begin{align*}
    h(A_0A_1)+h(A_2|A_1)+h(B_1) &\leq h^*(A_0A_1)+h^*(A_2|A_1)+h^*(B_1)=2\\
    h(B_0B_1)+h(B_2|B_1)+h(C_1)  &\leq h(B_0B_1)+h(B_2|B_1)+h(C_1)=2 \\
    h(A_1) + h(C_0C_1) + h(C_2|C_1)  &\leq h^*(A_1) + h^*(C_0C_1) + h^*(C_2|C_1)=2
  \end{align*}
  We can argue that no axis-parallel partition can cover the three
  inequalities above by using an argument similar to the hexagon
  query: the first chain inequality corresponds to joining the light
  $A_2$'s with the heavy $B_1$'s, and similarly for other two
  inequalities, and this leaves out, for example, the tuples where
  $A_1, B_1, C_1$ are all heavy, or all light.  On the other hand,
  every edge dominated polymatroid $h$ must satisfy at least one of
  these three chain inequalities, because their sum is always
  satisfied by any edge-dominated polymatroid:
  \begin{align*}
    h(A_0A_1)+&h(A_1A_2) + h(B_0B_1)+h(B_1B_2) +  h(C_0C_1)+h(C_1C_2) \\
    \leq &h^*(A_0A_1)+h^*(A_1A_2) + h^*(B_0B_1)+h^*(B_1B_2) + h^*(C_0C_1)+h^*(C_1C_2)=4
  \end{align*}
  Next, we convert the space partition into an algorithm, which
  computes each target using a single multi-join:
\begin{align*}
&   \text{for } (a_0,a_1) \text{ in } R_1(A_0,A_1), b_1 \text{ in } S_2(B_1): \\
&   \text{ if } \left(|R_1(A_0A_1)|\cdot \degree_{R_2}(A_2|a_1)\cdot \frac{|S_2(B_1B_2)|}{\degree_{S_2}(B_2|b_1)}\leq N^2\right) &\text{// or just } \degree_{R_2}(A_2|a_1) \leq & \degree_{S_2}(B_2|b_1) \\
&     \text{\ \ for } a_2 \text{ in } R_2(A_2|a_1): \\
&       \text{\ \ \ Output } U(a_0,a_1,a_2,b_1) \\
& \\
&   \text{for } (b_0,b_1) \text{ in } S_1(B_0,B_1), c_1 \text{ in } T_2(C_1):\\
&   \text{ if } \left(|S_1(B_0B_1)|\cdot \degree_{S_2}(B_2|b_1)\cdot \frac{|T_2(C_1C_2)|}{\degree_{T_2}(C_2|c_1)}\leq N^2\right) &\text{// or just } \degree_{S_2}(B_2|b_1) \leq & \degree_{T_2}(C_2|c_1) \\
&     \text{\ \  for } b_2 \text{ in } S_2(B_2|b_1):\\
&       \text{\ \ \ Output } V(b_0,b_1,b_2,c_1)\\
& \\
&   \text{for } (c_0,c_1) \text{ in } R(C_0,C_1), a_1 \text{ in } R_2(A_1):\\
&   \text{ if } \left(|T_1(C_0C_1)|\cdot \degree_{T_2}(C_2|c_1)\cdot \frac{|R_2(A_1A_2)|}{\degree_{R_2}(A_2|a_1)}\leq N^2\right)&\text{// or just } \degree_{T_2}(C_2|c_1) \leq & \degree_{R_2}(A_2|a_1) \\
&     \text{\ \ for } c_2 \text{ in } T_2(C_2|c_1):\\
&       \text{\ \ \ Output } W(c_0,c_1,c_2,a_1)
\end{align*}
The first query populates the target $U$ by first iterating over all
pairs $(a_0,a_1) \in R_1(A_0A_1)$ and $b_1 \in S_2(B_1)$ that satisfy
the condition
$|R_1(A_0A_1)|\cdot \degree_{R_2}(A_2|a_1)\cdot
\frac{|S_2(B_1B_2)|}{\degree_{S_2}(B_2|b_1)}\leq N^2$, which ensures
that $|U| \leq N^2$.  In fact, it suffices to check just
$\degree_{R_2}(A_2|a_1) \leq \degree_{S_2}(B_2|b_1)$, because this
implies
$|R_1(A_0A_1)|\cdot \degree_{R_2}(A_2|a_1)\cdot
\frac{|S_2(B_1B_2)|}{\degree_{S_2}(B_2|b_1)}\leq N^2$ by the fact that
$|R_1|\cdot |S_2| \leq N^2$. Similarly for the other two queries.
\end{example}

\begin{example} \label{ex:discussion:2}
  A more subtle example is the following DDR:
  \begin{align*}
    U(A_1A_2A_3A_4A_5)&\vee V(A_3A_4A_5A_6A_1) \vee W(A_5A_6A_1A_2A_3) \vee Z(A_2A_4A_6)\\
    =  &   R_1(A_1A_2A_3)\wedge R_2(A_2A_3A_4) \wedge R_3(A_3A_4A_5) \wedge R_4(A_4A_5A_6) \wedge R_5(A_5A_6A_1)\wedge R_6(A_6A_1A_2)
  \end{align*}
  Assuming all inputs have size $\leq N$, we show that one can compute
  an output where all targets have size $\leq N^{3/2}$.  For that, we
  consider the following good Shannon inequality:
  \begin{align*}
    & h(A_1A_2A_3A_4A_5)+h(A_3A_4A_5A_6A_1)+h(A_5A_6A_1A_2A_3)+h(A_2A_4A_6)\\
    & \leq   h(A_1A_2A_3)+h(A_2A_3A_4)+h(A_3A_4A_5)+h(A_4A_5A_6)+h(A_5A_6A_1)+h(A_6A_1A_2)
  \end{align*}
which we further write as a sum of the following chain inequalities:
  \begin{align*}
    h(A_1A_2A_3A_4A_5) &\leq h(A_1A_2A_3)+h(A_4|A_2A_3)+h(A_5|A_4)\\
    h(A_3A_4A_5A_6A_1) &\leq h(A_3A_4A_5)+h(A_6|A_4A_5)+h(A_1|A_6)\\
    h(A_5A_6A_1A_2A_3) &\leq h(A_5A_6A_1)+h(A_2|A_6A_1)+h(A_3|A_2)\\
    h(A_2A_4A_6)       &\leq h(A_2)+h(A_4)+h(A_6)
  \end{align*}
  The space of edge dominate polymatroids is partitioned into four
  sets, each defined by the constraint that the RHS of one of the
  chains above is $\leq 3/2$: each edge dominated polymatroid
  satisfies at least one of these constraints, because their sum
  asserts $h(A_1A_2A_3)+\cdots+h(A_6A_1A_2) \leq 6$, which is always
  satisfied.  Furthermore, we one can convert this partition of the
  space into an algorithm that computes separately each target
  $U, V, W, Z$, similarly to the previous example.  What makes this
  example more subtle is the fact that we need to first examine
  $\degree(A_4|A_2A_3)$ in $R_2$, then examine $\degree(A_3|A_2)$ in
  the projection of $R_2$ on $A_2A_3$.
\end{example}

\begin{example} \label{ex:discussion:3} Finally, we illustrate a DDR whose associated
  Shannon inequality is {\em not} a sum of chains, and where the dynamic partitioning of the
  data done by $\newalgorithm$ appears to be unavoidable.  The DDR is:

  \begin{align*}
    U(A_1A_2A_3A_4)\vee&V(B_1B_2B_3B_4)\vee W(A_1A_3B_1B_3)\vee Z_1(A_2B_2)\vee Z_2(A_4B_4) \\
    &= R_1(A_1A_2)\wedge R_2(A_2A_3)\wedge R_3(A_3A_4) \wedge R_4(A_4A_1) \\
    \wedge  & S_1(B_1B_2)\wedge S_2(B_2B_3)\wedge S_3(B_3B_4) \wedge S_4(B_4B_1)
  \end{align*}
  When all inputs are $\leq N$, then one can compute an output where
  all targets have size $\leq N^{8/5}$, because of the following
  Shannon inequality:
  \begin{align}
    h(A_1A_2A_3A_4)+&h(B_1B_2B_3B_4)+h(A_1A_3B_1B_3)+h(A_2B_2)+h(A_4B_4) \nonumber\\
    & \leq h(A_1A_2)+h(A_2A_3)+h(A_3A_4)+h(A_4A_1) \label{eq:aaaabbbb}\\
    + & h(B_1B_2)+h(B_2B_3)+h(B_3B_4)+h(B_4B_1) \nonumber
  \end{align}

\noindent To see that this inequality is valid, observe that it is the
  summation of the following 9 chain inequalities:
\begin{align*}
 \emblue{h(A_1A_2A_3)} &\leq h(A_1A_2)+h(A_3|A_2) \\
 \emblue{h(B_1B_2B_3)} &\leq h(B_1B_2)+h(B_3|B_2) \\
 \emblue{h(A_1A_3A_4)} &\leq h(A_3A_4)+h(A_1|A_4) \\
 \emblue{h(B_1B_3B_4)} &\leq h(B_3B_4)+h(B_1|B_4) \\
 h(A_1A_3B_1B_3) &\leq h(A_1A_3) + h(B_1B_3) \\
 h(A_1A_2A_3A_4) &\leq \emblue{h(A_1A_2A_3)} + \emblue{h(A_4|A_1A_3)}\\
 h(B_1B_2B_3B_4) &\leq \emblue{h(B_1B_2B_3)} + \emblue{h(B_4|B_1B_3)}\\
 h(A_2B_2) &\leq h(A_2)+h(B_2) \\
 h(A_4B_4) &\leq h(A_4)+h(B_4)
\end{align*}
However, inequality~\eqref{eq:aaaabbbb} is not strictly speaking a
``sum of chains'' as we defined it, because the blue terms in the
expressions above cancel out.  Concretely, this means that the
algorithm that follows these 9 chain inequalities must first compute
the blue intermediate results, then partition these again.  $\newalgorithm$
achieves this by keeping track of the probabilities in the temporary
results, corresponding to the blue terms.
\end{example}

Our discussion leads to the natural question of characterizing the
DDRs that can be computed like in Examples~\ref{ex:discussion:1}
and~\ref{ex:discussion:2}.  In the remainder of this section we
provide a sufficient condition.

Let $\bm V = \set{X_1, \ldots, X_n}$.  A \emph{Shearer inequality} is
a Shannon inequality of the form:
\begin{align}
  k\cdot h(\bm V) \leq & \sum_{i=1,m} h(\bm Z_i) \label{eq:shearer}
\end{align}
where $\bm Z_i \subseteq \bm V$ for $i=1,m$.  It is known that the
inequality is valid iff every variable $X_i \in \bm V$ occurs in at
least $k$ terms on the RHS.  We can assume w.l.o.g. that $X_i$ occurs
in exactly $k$ terms: otherwise, drop $X_i$ arbitrary for terms until
only $k$ copies remain.  Under this assumption we prove:

\begin{lmm}
  Shearer's inequality~\eqref{eq:shearer} is the sum of chain inequalities.
\end{lmm}

\begin{proof}
  We prove the following claim.  If $k \geq 1$ then there exists a
  chain inequality:
  \begin{align}
    h(\bm V) \leq & \sum_{j=1,\ell}h(\bm Y_j|\bm X_j) \label{eq:shearer:chain}
  \end{align}
  such that, for all $j=1,\ell$ there exists $i =1,m$ s.t.  $\bm X_j
  \bm Y_j=\bm Z_i$.  We first show that the claim implies the lemma.
  This follows by observing that the difference between the RHS
  for~\eqref{eq:shearer} and~\eqref{eq:shearer:chain} is the RHS of
  another Shearer inequality:
  \begin{align}
  (k-1)\cdot h(\bm V) \leq & \sum_{i=1,m} h(\bm Z_i') \label{eq:shearer:2}
  \end{align}
  In other words, the sum of inequalities~\eqref{eq:shearer:chain}
  and~\eqref{eq:shearer:2} is the inequality~\eqref{eq:shearer}, and
  the lemma follows by applying the claim repeatedly until $k=0$.

  To prove the claim, we prove by induction on $p$ that there exists a
  chain inequality:
  \begin{align*}
    h(\bm W) \leq & \sum_{j=1,\ell}h(\bm Y_j|\bm X_j)
  \end{align*}
  where $|\bm W|=p$ and each set $\bm X_j\bm Y_j$ is equal to some
  term $\bm Z_i$ in~\eqref{eq:shearer}.  When $p=0$ then we set
  $\ell=0$ and the claim holds vacuously.  Assuming we have such a
  chain for $p\geq 0$, consider some variable $V \not\in \bm W$.
  There exists some term $h(\bm Z)$ in~\eqref{eq:shearer} that
  contains $V$, and we partition $\bm Z$ into $\bm Z = \bm X \bm Y$
  where $\bm X \defeq \bm Z\cap (\bigcup_{j=1,\ell} \bm Y_j)$ and
  $\bm Y \defeq \bm Z - \bm X$.  Then we obtain the longer chain
  inequality:
  \begin{align*}
    h(\bm W \bm Y) \leq & \sum_{j=1,\ell}h(\bm Y_j|\bm X_j) + h(\bm Y|\bm X)
  \end{align*}
  which completes the proof.
\end{proof}

\end{document}